\documentclass[11pt]{article}
\pdfoutput=1
\usepackage[T1]{fontenc}
\usepackage{lmodern}
\usepackage{amsthm}
\usepackage{graphicx} 
\usepackage{array} 

\usepackage{amsmath, amssymb, amsfonts, verbatim}
\usepackage{hyphenat,subfigure,multirow}
\usepackage{mathtools}
\usepackage{stmaryrd}
\usepackage[usenames,dvipsnames]{xcolor}

\usepackage{dsfont}
\usepackage{amstext}
\usepackage{setspace}
\usepackage{fancyhdr}
\usepackage{lastpage}
\usepackage{extramarks}
\usepackage{chngpage}
\usepackage{soul,color}
\usepackage{float,wrapfig}
\usepackage{CJK}
\usepackage{makecell}
\usepackage{longtable}
\usepackage[ruled, linesnumbered]{algorithm2e}
\usepackage{listings}

\definecolor{DarkRed}{rgb}{0.5,0.1,0.1}
\definecolor{DarkBlue}{rgb}{0.1,0.1,0.5}

\usepackage{hyperref}
\hypersetup{
colorlinks=true,
pdfnewwindow=true,
citecolor=ForestGreen,
linkcolor=DarkRed,
filecolor=DarkRed,
urlcolor=DarkBlue
}

\usepackage{bm}
\usepackage{url}
\usepackage{xspace}
\usepackage[mathscr]{euscript}

\usepackage{mdframed}

\usepackage{cite}
\usepackage{enumitem}

\usepackage[margin=1in]{geometry}

\newtheorem{theorem}{Theorem}
\newtheorem{lemma}{Lemma}

\newtheorem{corollary}{Corollary}

\newtheorem*{definition*}{Definition}

\theoremstyle{definition}
\newtheorem{remark}{Remark}
\newtheorem{definition}{Definition}

\newtheorem*{claim*}{Claim}
\newtheorem*{proposition*}{Proposition}
\newtheorem*{lemma*}{Lemma}
\newtheorem*{problem*}{Problem}

\newtheorem{mdresult}[theorem]{Theorem}

\newcommand{\ignore}[1]{}

\newcommand{\ceil}[1]{\left\lceil{#1}\right\rceil}
\newcommand{\floor}[1]{\left\lfloor{#1} \right\rfloor}
\newcommand{\bydef}{\stackrel{\mathrm{def}}{=}}
\newcommand{\poly}{\operatorname{poly}}

\newcommand{\abs}[1]{\left|#1\right|} 
\newcommand{\Side}[1]{\operatorname{Side}_{#1}}
\newcommand{\Region}[1]{\operatorname{Region}_{#1}}
\newcommand{\Match}[1]{\operatorname{Match}_{#1}}
\newcommand{\Cuts}[1]{\operatorname{Cuts}_{#1}}
\newcommand{\MinCut}[1]{\operatorname{MinCut}_{#1}}
\newcommand{\Small}[1]{\operatorname{Sm}_{#1}}

\newcommand{\FindSmall}[0]{\mathsf{FindSmall}}
\newcommand{\Update}[0]{\mathsf{Update}}
\newcommand{\Expand}[0]{\mathsf{Expand}}

\newcommand{\cir}{$^{\circ}$}

\makeatletter
\newcommand{\rmnum}[1]{\romannumeral #1}
\newcommand{\Rmnum}[1]{\expandafter\@slowromancap\romannumeral #1@}
\makeatother

\allowdisplaybreaks

\title{The Structure of Minimum Vertex Cuts\thanks{This work was supported by NSF grants CCF-1637546 and CCF-1815316.}}

\author{Seth Pettie\\ University of Michigan \\ \footnotesize
\texttt{pettie@umich.edu}
\and
Longhui Yin\\ Tsinghua University\\
\footnotesize \texttt{ylh17@mails.tsinghua.edu.cn}
}

\date{}

\begin{document}
\maketitle

\begin{abstract}
In this paper we continue a long line of work 
on representing the \emph{cut structure} of graphs.  
We classify the types minimum \emph{vertex} cuts, 
and the possible relationships between multiple minimum vertex cuts.

As a consequence of these investigations, we exhibit a
simple $O(\kappa n)$-space data structure that can quickly 
answer pairwise $(\kappa+1)$-connectivity queries in a 
$\kappa$-connected graph.  We also show how to compute the 
``closest'' $\kappa$-cut to every vertex in near linear $\tilde{O}(m+\poly(\kappa)n)$ time.

\end{abstract}

\thispagestyle{empty}
\setcounter{page}{0}

\newpage

\section{Introduction}

One of the strong themes running through graph theory is to understand the \emph{cut structure} of graphs and to apply these structural theorems to solve algorithmic and data structural problems.
Consider the following exemplars of this line of work:
\begin{description}
  \item[Gomory-Hu Tree.] Gomory and Hu (1961)~\cite{GomoryH61}
  proved that any weighted, undirected graph 
  $G=(V,E)$ can be replaced by a weighted, undirected tree
  $T=(V,E_T)$ such that for every $s,t\in V$, 
  the minimum $s$-$t$ cut partition in $T$ 
  (removing a single edge, partitioning $V$ into two sets) 
  corresponds to a minimum $s$-$t$ cut partition in $G$.  These
  are sometimes called \emph{cut-equivalent trees}~\cite{AbboudKT20}.
  \item[Cactus Representations.] Dinitz, Karzanov, and Lomonosov (1976)~\cite{DinicKL76} proved that all the \emph{global} 
  minimum edge-cuts of any weighted, undirected graph $G=(V,E)$ 
  could be succinctly encoded as an (unweighted) \emph{cactus graph}.
  A cactus is a multigraph in which every edge participates in 
  exactly one cycle.  It was proved that there exists a cactus
  $C=(V_C,E_C)$ and an embedding $\phi : V\rightarrow V_C$ such
  that the minimum edge-cuts in $C$ (2 edges in a common cycle)
  are in 1-1 correspondence with the minimum edge-cuts of $G$.
  A corollary of this theorem is that there are at most $n \choose 2$
  minimum edge-cuts.
  \item[Picard-Queyrenne Representation.] In a \emph{directed} $s$-$t$
  flow network there can be exponentially many min $s$-$t$ cuts.
  Picard and Queyrenne (1980)~\cite{PicardQ80} proved that the family 
  $\mathscr{S}=\{S \mid (S,\overline{S}) \mbox{ is a min $s$-$t$}\}$ 
  corresponds 1-1 with the downward-closed sets of a partial order,
  and is therefore closed under union and intersection.
  \item[Block Trees, SPQR Trees, and Beyond.] Whitney (1932)~\cite{Whitney32a,Whitney32b}
  proved that the cut vertices (articulation points) of an undirected graph $G=(V,E)$ partition $E$ into single edges and 2-edge connected components (blocks).  This yields the \emph{block tree} representation.  Di Battista and Tamassia (1989)~\cite{BattistaT89,DiBattistaT96} formally defined the \emph{SPQR tree}, which succinctly encodes all 2-vertex cuts in a biconnected graph,
  and Kanevsky, Tamassia, Di Battista, and Chen~\cite{KanevskyTBC91}
  extended this structure to represent 3-vertex cuts in a triconnected 
  graph.\footnote{Many of the structural insights behind~\cite{DiBattistaT96,KanevskyTBC91} were latent 
  in prior work. See, for example. Mac Lane~\cite{MacLane37} (1937), 
  Tutte~\cite{Tutte61,Tutte66} (1961-6), Hopcroft and Tarjan~\cite{HopcroftT73},
  and Cunningham and Edmonds~\cite{CunninghamE80}.}
\end{description}

\begin{figure}
    \centering

\begin{tabular}{c@{\hspace{1.5cm}}c}
\scalebox{.5}{\includegraphics{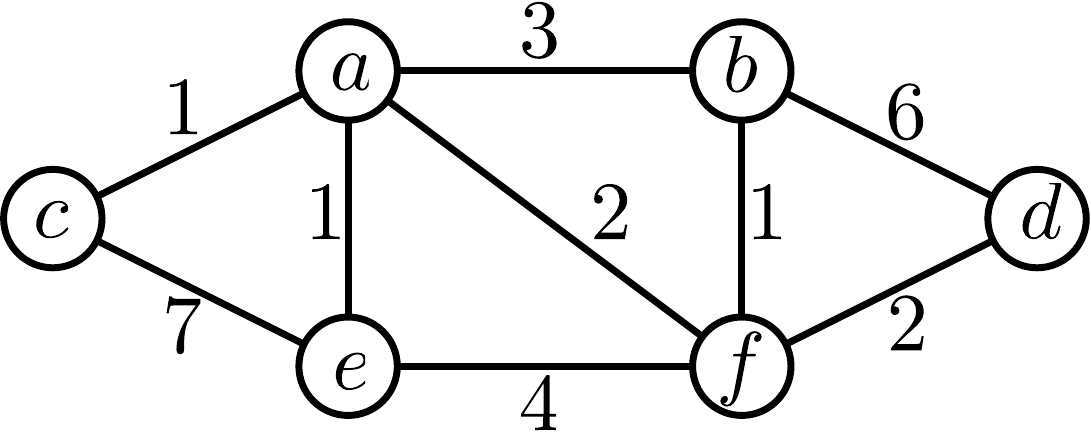}} 
&  
\scalebox{.5}{\includegraphics{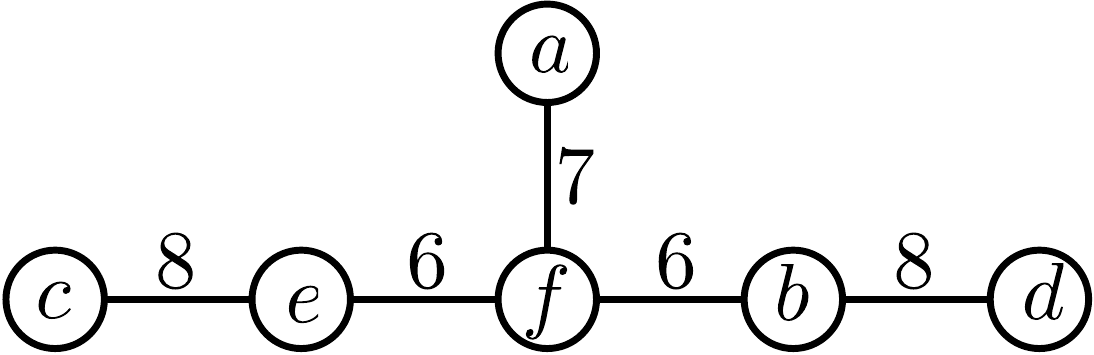}} 
\\
(a) & (b)\\
&\\
\scalebox{.5}{\includegraphics{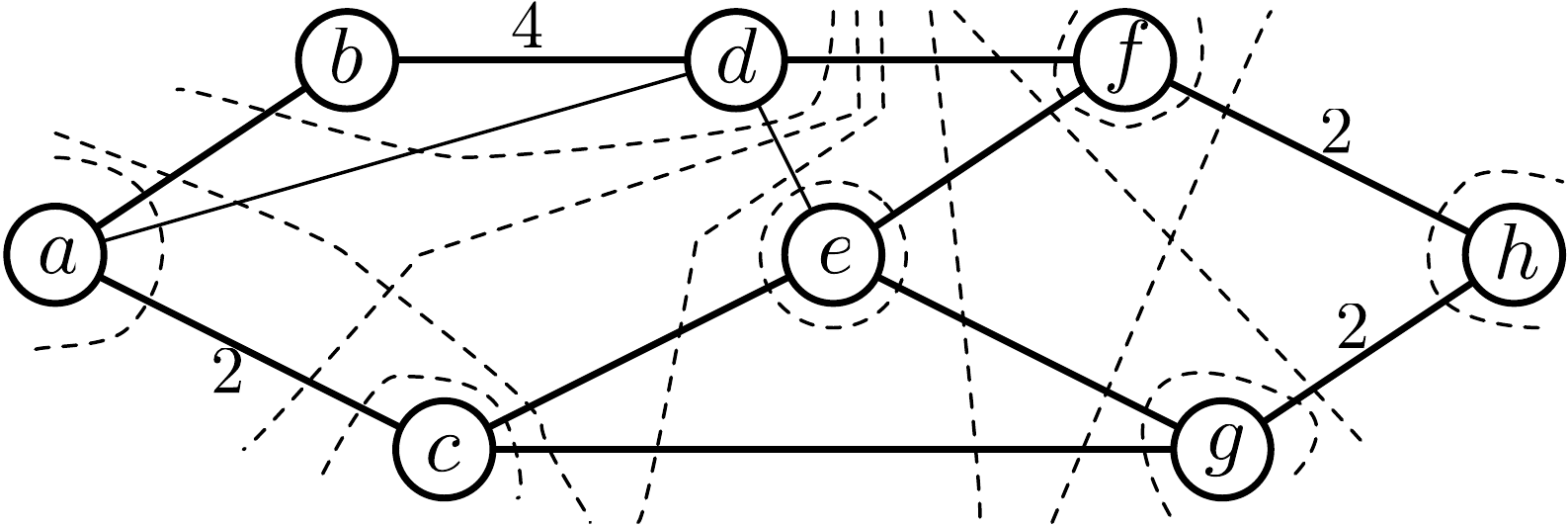}} 
&  
\scalebox{.5}{\includegraphics{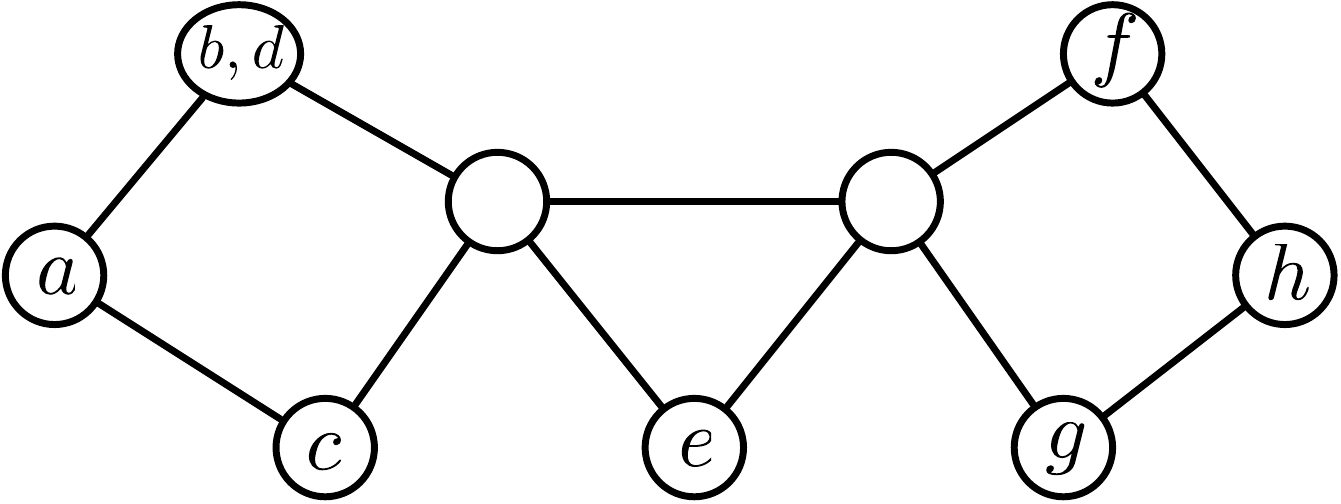}} 
\\
(c) & (d)\\
&\\
\scalebox{.5}{\includegraphics{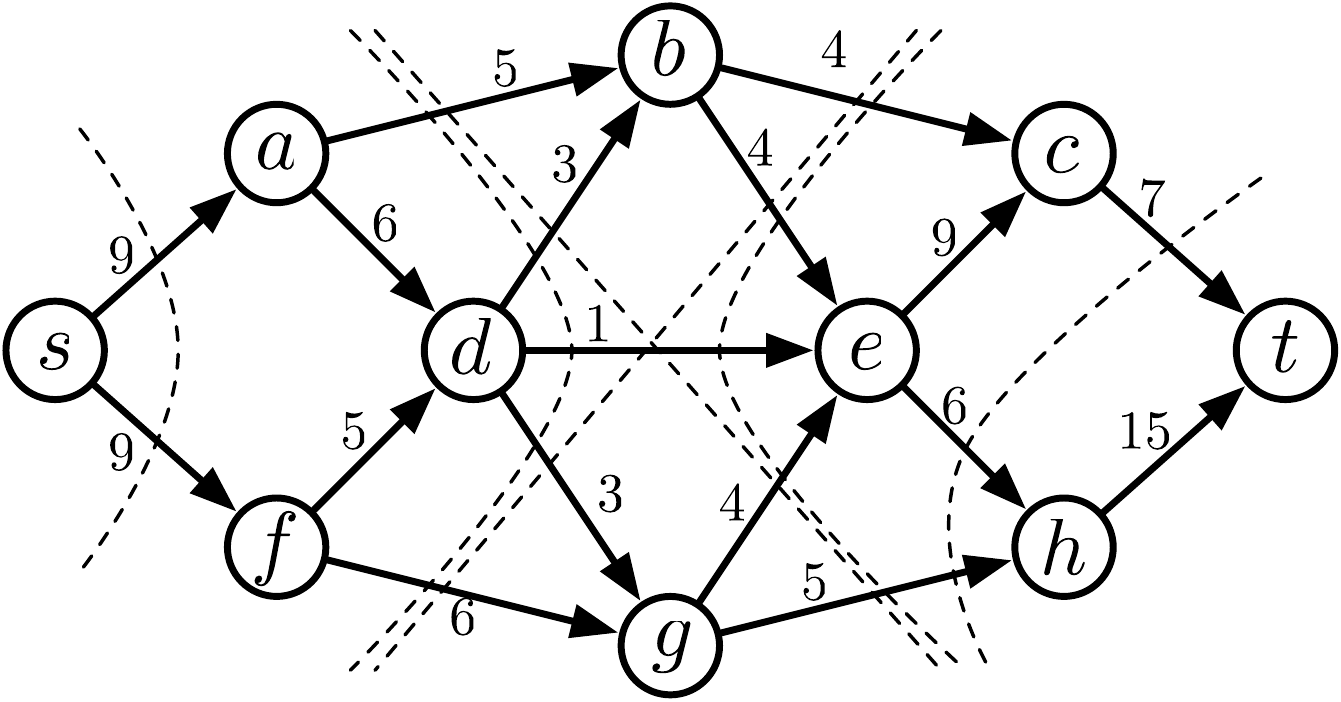}} 
&  
\scalebox{.5}{\includegraphics{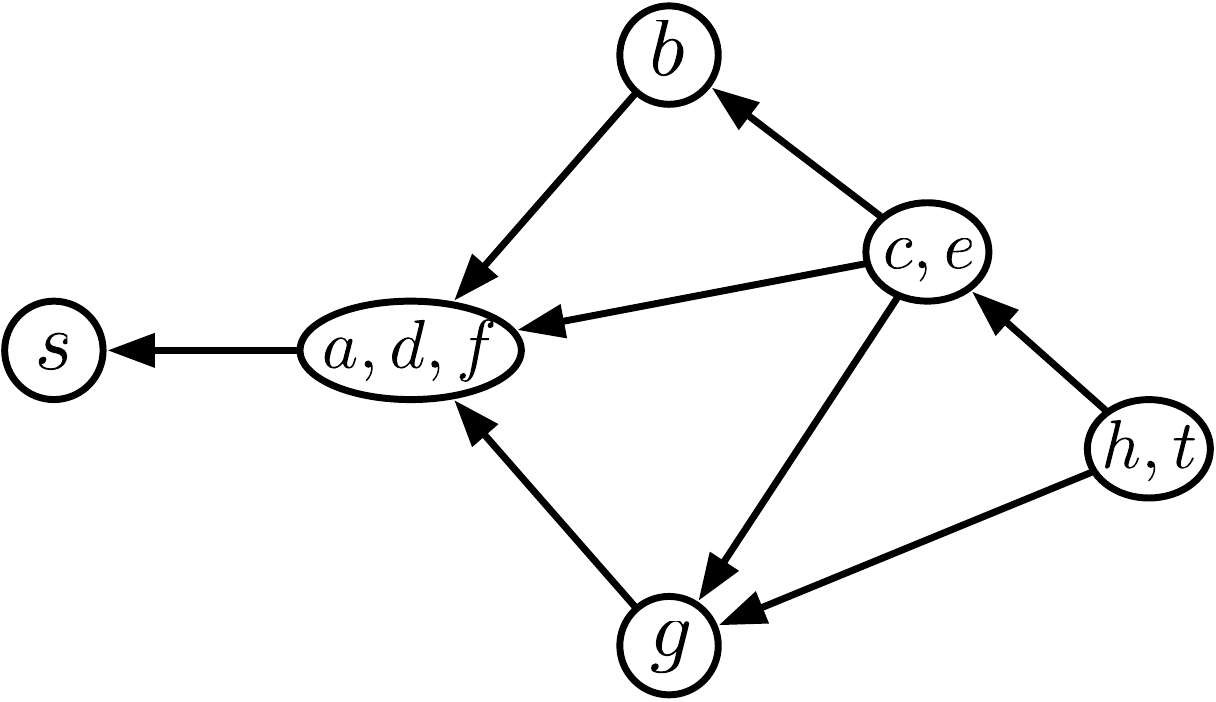}} 
\\
(e) & (f)\\
&\\
\scalebox{.45}{\includegraphics{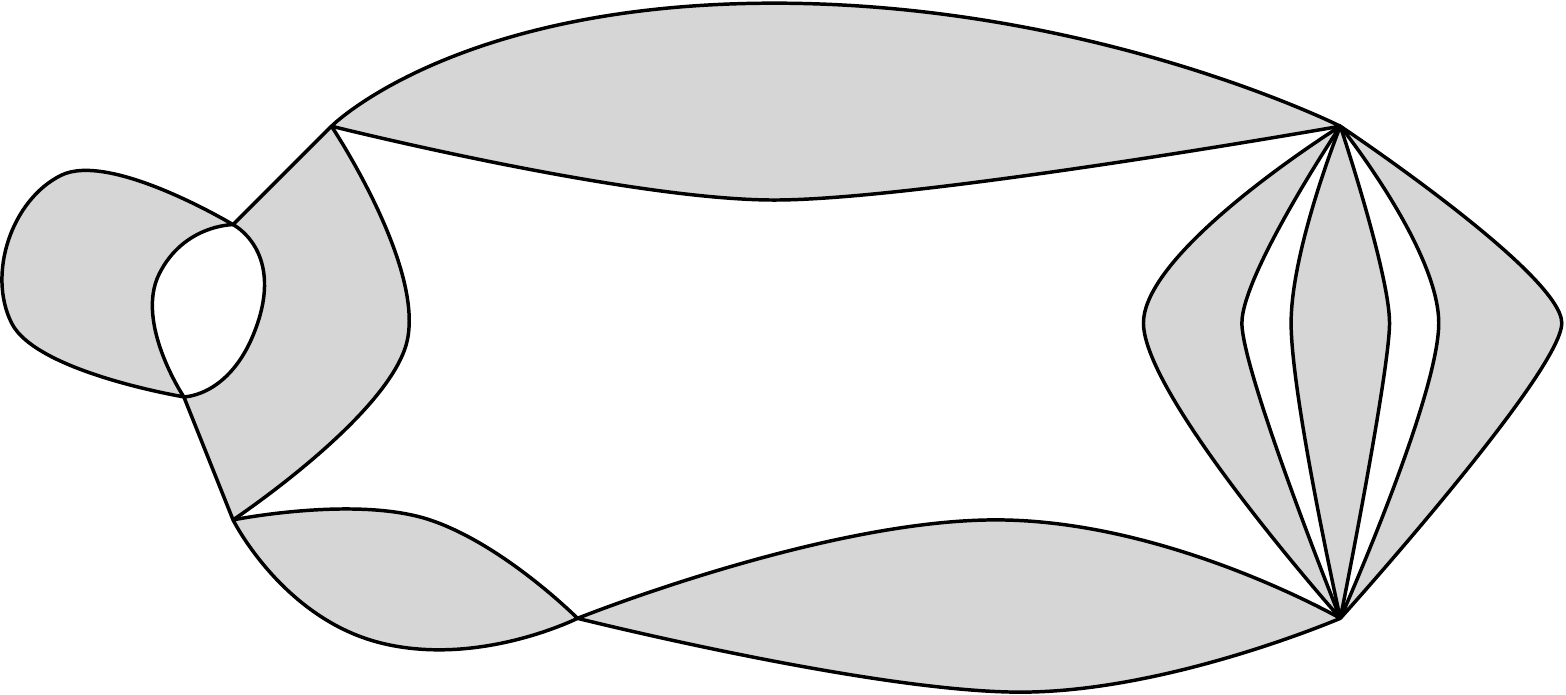}} 
&  
\scalebox{.45}{\includegraphics{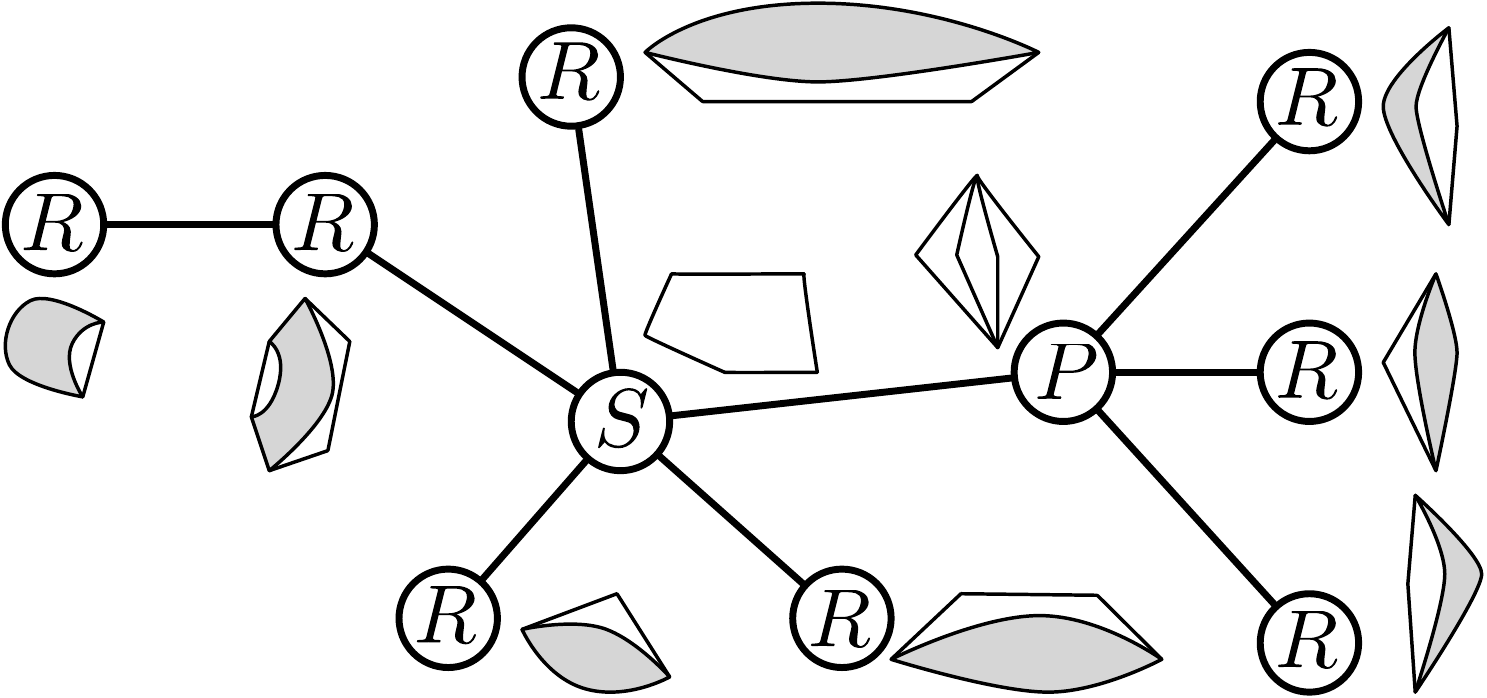}} 
\\
(g) & (h)\\
\end{tabular}
    \caption{(a) A weighted undirected graph; (b) Its Gomory-Hu (cut-equivalent) tree~\cite{GomoryH61}.
    (c) A weighted undirected graph (unmarked edges have unit weight); (d) the Cactus representation~\cite{DinicKL76} of its minimum edge cuts.
    (e) A directed $s$-$t$ flow network; (f) A dag whose downward-closed sets (that include $s$ but not $t$) 
    correspond to min $s$-$t$ cuts (Picard-Queyrenne~\cite{PicardQ80}).
    (g) An abstract representation of a 2-connected graph; (h) The representation of its 3-connected components as an SPQR tree (Di Battista-Tamassia~\cite{DiBattistaT96}).}
    \label{fig:cut-structures}
\end{figure}

It is natural to ask how, and to what extent, 
these structures can be extended and generalized.  
Gusfield and Naor~\cite{GusfieldN90} 
described an analogue of Gomory-Hu trees for vertex connectivity, 
i.e., a tree that compactly represents the $s$-$t$ \underline{\emph{vertex}} 
connectivity for every $s,t\in V$.  It used a result of 
Schnorr~\cite{Schnorr79} on an analogue of Gomory-Hu 
trees for ``roundtrip'' flow-values in directed networks.
These claims were \underline{refuted} by Benczur~\cite{Benczur95a},
who illustrated that Schnorr's and Gusfield and Naor's 
proofs were incorrect and could not be rectified.
In particular, $s$-$t$ vertex connectivity and directed $s$-$t$ cuts 
have \emph{no tree representation}.  We take this as a reminder
that having published proofs (\emph{even incorrect ones}) is 
essential for facilitating self-correction in science.

\medskip 

The inspiration for this paper is an extended 
abstract of Cohen, Di Battista, Kanevsky, and Tamassia~\cite{cohen1993reinventing} from STOC 1993.
Their goal was to find a cactus-analogue 
for global minimum vertex cuts, or from a different perspective, to extend SPQR trees~\cite{DiBattistaT96} and~\cite{KanevskyTBC91} from $\kappa \in\{2,3\}$ vertex cuts to arbitrarily large $\kappa$.  
As an application of their ideas, 
they described a data structure for $\kappa$-connected graphs
occupying space $O(\kappa^3 n)$ that, given $u,v$, 
decided whether 
$u,v$ are separated by a $\kappa$-cut or $(\kappa+1)$-connected.  
There are no suspect claims in~\cite{cohen1993reinventing}.
On the other hand, the paper is 7 pages and leaves many of its 
central claims unproven.\footnote{The full version of this 
paper was never written (personal communication with 
R. Tamassia, 2011, and R. Di Battista, 2016).}
We believe that understanding the \emph{structure} of minimum 
vertex cuts is a fundamental problem in graph theory, 
and deserving of a complete, formal treatment.

In this paper we investigate the structure of 
the set of all minimum vertex cuts and classify 
the relationships between different minimum vertex cuts.
Our work reveals some structural features of minimum $\kappa$-cuts 
not evident in Cohen, Di Battista, Kanevsky, and Tamassia~\cite{cohen1993reinventing}, and 
ultimately allows us to develop a simpler data 
structure to answer pairwise $\kappa$-cut queries in 
a $\kappa$-connected graph.  It occupies (optimal) $O(\kappa n)$ 
space and can be constructed in randomized 
$\tilde{O}(m + \poly(\kappa)n)$ time, in contrast to~\cite{cohen1993reinventing}, 
which occupies $O(\kappa^3 n)$ space and is 
constructed in $\exp(\kappa)n^5$ time.\footnote{The algorithm 
enumerates \emph{all} minimum $\kappa$-cuts, which can be as large
as $\Omega(2^\kappa (n/\kappa)^2)$; modern vertex connectivity 
algorithms~\cite{ForsterNYSY20,GaoLNPSY19,Gabow06} may reduce
the exponent of $n$ in the running time.}

\subsection{Related Work}

Dinitz and Vainshtein~\cite{DinitzV94,DinitzV95} 
combined elements of the cactus~\cite{DinicKL76} and 
Picard-Queyrenne~\cite{PicardQ80} representations, which they called the \emph{connectevity carcass}.  
Given an undirected,
unweighted $G=(V,E)$ and $S\subseteq V$ of terminals, 
$\lambda_S$ is the size of the minimum edge-cut that separates $S$.  The carcass represents
\emph{all} size-$\lambda_S$ separating cuts in $O(\min\{m, \lambda_S n\})$ space and answers
various cut queries in $O(1)$ time.\footnote{The carcass was introduced in extended abstracts~\cite{DinitzV94,DinitzV95}
and the (simpler) case of odd $\lambda_S$ was analyzed in detail in a journal article~\cite{DinitzV00}.  
We are not aware of a full treatment of the case when $\lambda_S$ is even.}

Benczur and Goemans~\cite{BenczurG08} generalized the cactus representation~\cite{DinicKL76}
in a different direction, by giving a compact representation 
of all cuts that are within
a factor $6/5$ of the global minimum edge-cut.

Dinitz and Nutov~\cite{DinitzN95} generalized the cactus representation~\cite{DinicKL76} 
in another direction, by giving an $O(n)$-space representation of all 
$\lambda$ and $\lambda+1$ edge cuts, where $\lambda$ is the edge-connectivity of the undirected, unweighted graph.
Unpublished manuscripts~\cite{DinitzN99a,DinitzN99b} give detailed treatments 
of the $\lambda$ odd and $\lambda$ even cases separately.

Georgiadis et al.~\cite{GeorgiadisIP20,FirmaniGILS16,GeorgiadisILP16,GeorgiadisILP18} investigated
various notions of 1- and 2-edge and vertex connectivity in \emph{directed} graphs, and the compact
representation of edge/vertex cuts.

\paragraph{Sparsification.} One general way to compactly represent connectivity
information is to produce a \emph{sparse} graph with the same cut structure.
Nagamochi and Ibaraki~\cite{NagamochiI92} proved that every unweighted, undirected
graph $G=(V,E)$ contains a subgraph $H=(V,E_H)$ with $|E_H|<(k+1)n$ such that
$H$ is computable in $O(m)$ time and
contains exactly the same $k'$-vertex cuts and $k'$-edge cuts 
as $G$, for all $k'\in\{1,\ldots,k\}$.
Benczur and Karger~\cite{BenczurK15} proved that for 
any capacitated, undirected graph $G=(V,E)$, there is another capacitated
graph $H=(V,E_H)$ with $|E_H|=O(\epsilon^{-2}n\log n)$ such that the 
capacity of \emph{every} cut in $G$ is preserved in $H$ 
up to a $(1\pm \epsilon)$-factor.  This bound was later improved to
$O(\epsilon^{-2}n)$ by Batson, Spielman, and Srivastava~\cite{BatsonSS12}, 
which is optimal.

In directed graphs, Baswana, Choudhary, and Roditty~\cite{BaswanaCR16} considered 
the problem of finding a sparse subgraph that preserves reachability
from a single source, even if $d$ vertices are deleted.
They proved that $\Theta(2^d n)$ edges are necessary and 
sufficient for $d\in [1,\log n]$.

\paragraph{$d$-Failure Connectivity.} An undirected graph can be compactly
represented such that connectivity queries can be answered after the
deletion of any $d$ vertices/edges (where $d$ could be much larger than
the underlying connectivity of the graph).  
Improving on~\cite{PatrascuT07,KapronKM13,DuanP10},
Duan and Pettie~\cite{DuanP20} proved that $d$ vertex failures
could be processed in $\tilde{O}(d^2)$ time such that connectivity 
queries are answered in $O(d)$ time, and $d$ edge failures could
be processed in $O(d\log d\log\log n)$ time such that 
connectivity queries are answered in $O(\log\log n)$ time.
The size of
the~\cite{DuanP20} structure is $\tilde{O}(m)$ for vertex failures and $\tilde{O}(n)$ for edge failures.
Choudhary~\cite{Choudhary16} gave an optimal $O(n)$-space data structure
that could answer directed reachability queries after $d\in\{1,2\}$
vertex or edge failures.

\paragraph{Labeling Schemes.} Benczur's refutation~\cite{Benczur95a}
of~\cite{Schnorr79,GusfieldN90} shows that all pairwise vertex connectivities
cannot be captured in a \underline{\emph{tree}} structure, 
but it does not preclude other representations
of this information.  Korman~\cite{Korman10} proved that the vertices of
any undirected $G=(V,E)$ could be assigned $O(k^2\log n)$-bit labels such that
given $(\operatorname{label}(u),\operatorname{label}(v))$, we can determine whether
$u$ and $v$ are $(k+1)$-connected or separated by a $k$-vertex cut.  That is,
$\poly(\overline{\kappa})\log n$-bit labels suffice to compute $\min\{\kappa(u,v),\overline{\kappa}\}$, where $\kappa(u,v)$ 
is the pairwise connectivity of $u,v$.

\paragraph{Vertex Connectivity Algorithms.}
In optimal linear time we can decide whether the connectivity
of a graph is $\kappa=1,\kappa=2,$ or $\kappa\ge 3$~\cite{Tar72,HopcroftT73}.
For larger $\kappa$, 
the state-of-the-art in vertex connectivity has been improved substantially in the last few years.  Forster, Nanongkai, Yang, Saranurak, and Yingchareonthawornchai~\cite{ForsterNYSY20} 
gave a \emph{Monte Carlo} algorithm for computing the vertex 
connectivity $\kappa$ of an undirected graph 
in $\tilde{O}(m + n\kappa^3)$ time, w.h.p.\footnote{The algorithm does not produce a witness, and hence may err with small probability.}
The best deterministic algorithm, due to Gao, Li, Nanongkai, Peng, Saranurak, and Yingchareonthawornchai~\cite{GaoLNPSY19},
computes the connectivity $\kappa < n^{1/8}$ in
$O((m+n^{7/4}\kappa^{O(\kappa)})n^{o(1)})$ time or
$O((m+n^{19/20}\kappa^{5/2})n^{o(1)})$ time.
For $\kappa > n^{1/8}$, Gabow's algorithm~\cite{Gabow06} 
runs in
$O(\kappa n^2 + \kappa^2 n\cdot \min\{n^{3/4}, \kappa^{3/2}\})$ time.

\subsection{Organization}
In Section~\ref{sect:preliminaries} we review basic definitions and lemmas regarding vertex cuts.
Section~\ref{sect:classification} gives
the basic classification theorem for minimum vertex cuts,
and lists some useful corollaries.
In short, every pair of cuts have \emph{laminar},
\emph{wheel}, \emph{crossing matching}, or \emph{small}
relation.  Sections~\ref{subsect:wheel}--\ref{subsect:small} analyze these four categories in more detail.
Section~\ref{sect:data-structure} exhibits a new $O(\kappa n)$-space data structure that, given two vertices, 
answers $(\kappa+1)$-connectivity queries in $O(1)$ time,
and produces a separating $\kappa$-cut (if one exists)
in $O(\kappa)$ time.
We conclude with some remarks and open problems in Section~\ref{sect:conclusion}.

\section{Preliminaries}\label{sect:preliminaries}

The input is a simple, connected, undirected graph $G=(V,E)$ with $n=\abs{V}$ and $m=\abs{E}$.  The predicate $A\subset B$ is true if 
$A$ is a \emph{strict} subset of $B$.

Let the subgraph of $G$ induced by $A$ be denoted $G|_{A}$. 
We call $U\subset V$ a \emph{cut} if the 
graph $G|_{V\backslash U}$ 
is disconnected. 
A \emph{side} of the cut $U$ is a connected component of 
$G|_{V\backslash U}$. 
If $P$ is a side of $U$ and $A\subseteq P$, 
we say $A$ is \emph{within a side of $U$}, 
and let $\Side{U}(A)=P$ denote the side containing $A$.
A \emph{region} of a cut $U$ is a side, 
or the union of several sides of $U$. Denote $\Region{U}(A)$ as the region containing the sides of $U$ that intersects with $A$.
\footnote{Note when $A$ is a singleton set $\{u\}$, $\Region{U}(A)=\Side{U}(A)$.}
We say a cut \emph{disconnects} or \emph{separates} 
$A$ and $B$ if they are in distinct sides of $U$.
In particular, 
if $B=V\backslash(A\cup U)$, we say $U$ \emph{disconnects} or \emph{separates}
$B$ \emph{from the rest of the graph}.

A path $\pi = v_1v_2\cdots v_l$ is \emph{from} $A$ \emph{to} $B$, 
if $v_1\in A$ and $v_l\in B$.  Two paths $\pi,\pi'$ from $v_1$ to $v_l$
are \emph{internally} vertex disjoint if they have no common 
vertices, except for $v_1,v_l$.
We say $U$ \emph{blocks} $\pi$ if $U\cap \{v_2,\ldots,v_{l-1}\}\neq \emptyset$.

A \emph{$k$-cut} is a cut of size $k$.
Define $\kappa(u,v)$ to be the minimum $k$ such 
that there exists a $k$-cut separating $u$ and $v$, 
where $\{u,v\} \neq E(G)$.
Define $\kappa = \kappa(G)$ to be the minimum of $\kappa(u,v)$ over all 
pairs $\{u,v\} \in {V(G)\choose 2}\backslash E(G)$. 
We say $G$ is \emph{$k$-connected} if $\kappa(G) \geq k$.

In this paper we assume that $\kappa < n/4$ and consider 
the set of \emph{all} (minimum) $\kappa$-cuts.

\begin{remark}\label{remark:kappa-def}
There is some flexibility in defining the corner cases.
Some authors leave $\kappa(u,v)$ undefined when $\{u,v\}\in E(G)$ 
or define it to be $n-1$.  In~\cite{cohen1993reinventing} a $k$-cut
is defined to be a mixed set of edges and vertices whose removal 
disconnects the graph.  Under this definition, 
when $\{u,v\}\in E(G)$, $\kappa(u,v)=k$ if removing $k-1$ vertices 
and $\{u,v\}$ disconnects $u$ and $v$.  This last definition
is compatible with Menger's theorem, and allows for it to be extended
to \emph{all} pairs of vertices.
\end{remark}

\begin{theorem}\label{thm:Menger} (Menger~\cite{Menger27})
Let $G=(V,E)$ be an undirected graph and $\{u,v\}$ a pair \emph{not} in $E$.
Let $U\subset V$ be a minimum size cut disconnecting $u$ and $v$
and $\Pi$ be a maximum size 
set of internally vertex disjoint paths
from $u$ to $v$.  Then $\kappa(u,v)=|U|=|\Pi|$.
\end{theorem}

The following categories make sense when applied to \emph{non}-minimal vertex cuts, but we are only interested in applying them to 
\emph{minimum} vertex cuts. Henceforth \emph{cut} 
usually means \emph{minimum cut}.

\paragraph{Laminar Cuts.} 
Let $U$ be a cut and $P$ be a side of $U$. 
If $W$ is a cut and $W\subset U\cup P$, 
we say $W$ is a \emph{laminar cut} of $U$ 
in side $P$.\footnote{These are sometimes called \emph{parellel cuts}.}

\begin{figure}[h!]
\centering
\scalebox{.4}{\includegraphics{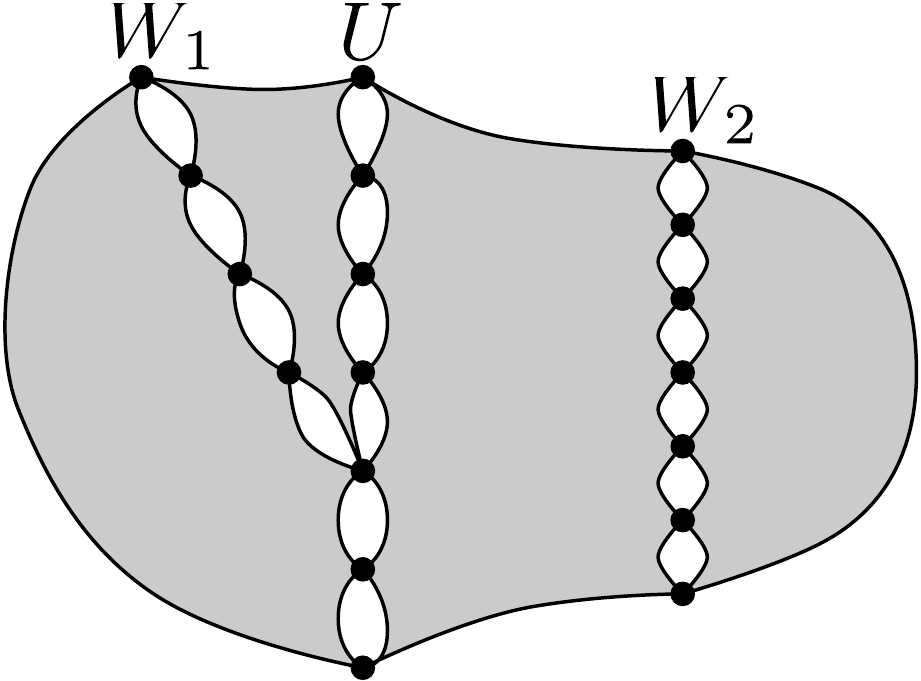}}
\caption{\label{fig:laminar}A 7-cut $U$ with
two sides, and two 7-cuts $W_1,W_2$ that are laminar w.r.t.~$U$.}
\end{figure}

\paragraph{Small Cuts.}
Informally, when a side of a cut tiny
we call the cut \emph{small}.
We define three levels of small cuts. 
Let $U$ be a cut with sides 
$A_1, A_2, \ldots, A_a$. 
We say that

\begin{itemize}
\item[1\cir] $U$ is (\Rmnum{1}, $t$)-\emph{small}
if there exists an index $i^{\sharp}$ such that 
$\sum_{i\neq i^{\sharp}}\abs{A_i} \leq t$. 
$A_{i^{\sharp}}$ is called the \emph{large side} of $U$ and
the others the \emph{small sides} of $U$.

\item[2\cir] $U$ is (\Rmnum{2}, $t$)-\emph{small} 
if there exists $i^{\sharp}$ such that 
for every $i\neq i^{\sharp}$, $\abs{A_i}\leq t$.

\item[3\cir] $U$ is (\Rmnum{3}, $t$)-\emph{small}, 
if there exists $i^{\sharp}$ such that $\abs{A_{i^{\sharp}}}\leq t$. 
In this case $A_{i^{\sharp}}$ is the \emph{small side} of $U$.
\end{itemize}

Note that for any $t$, 
\Rmnum{1}-small cuts are \Rmnum{2}-small, 
and \Rmnum{2}-small cuts are \Rmnum{3}-small.
We typically apply this definition with $t=\kappa$,
$t=\Theta(\kappa)$, or $t=\ceil{\frac{n-\kappa}{2}}$.

\paragraph{Wheel Cuts.} 
Suppose $V$ can be partitioned into a series of disjoint sets 
$T$, $\{C_i\}, \{S_i\}$ 
($1\leq i\leq w, w\geq 4$, subscripts are taken module $w$), 
such that the $\{C_i\}$ and $\{S_i\}$ are 
nonempty ($T$ may be empty), 
and $C_i\cup T\cup C_{i+2}$ disconnects $S_i\cup C_i\cup S_{i+1}$ 
from the rest of the graph. 
We say $(T;C_1, C_2, \ldots, C_w)$ forms a \emph{$w$-wheel} 
with \emph{sectors} $S_1, S_2, \ldots, S_w$. 
We call $T$ the \emph{center} of the wheel, 
$\{C_i\}$ the \emph{spokes} of the wheel, 
and $C(i,j) = C_i\cup T\cup C_j$ the \emph{cuts} of the wheel.
Define $D(i, j)=S_i\cup C_{i+1}\cup \cdots \cup C_{j-1} \cup S_{j-1}$.

Recall that we are only interested in wheels whose cuts 
are minimum $\kappa$-cuts.
The cut of the wheels discussed in this paper are all $\kappa$-cuts. 
It is proved in Lemma~\ref{lem3} that, if $(T;C_1, C_2, \ldots, C_w)$
forms a wheel, then for every $i, j$ such that $j-i\notin\{1, w-1\}$, $C(i, j)$ is a $\kappa$-cut with exactly two sides, 
namely $D(i, j)$ and $D(j, i)$.
\begin{figure}[h!]
    \centering
    \scalebox{.6}{\includegraphics{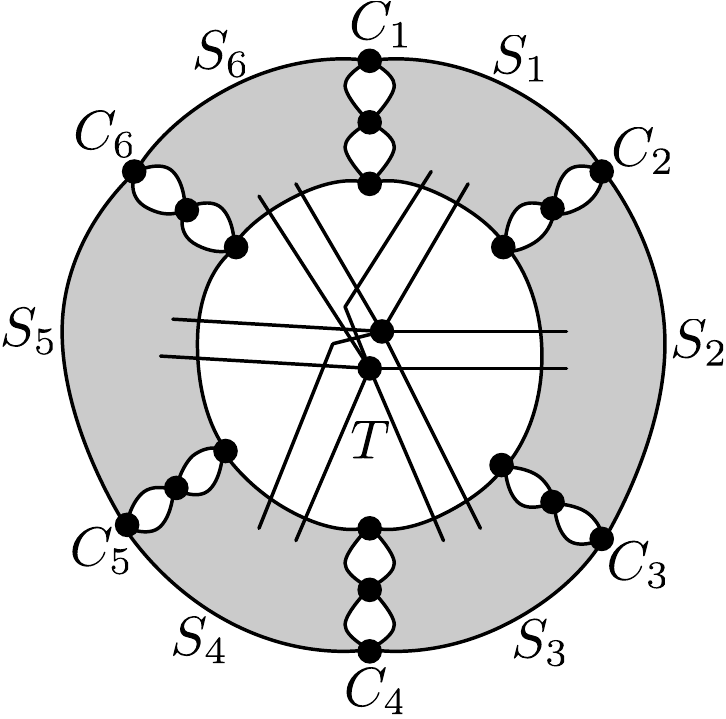}}
    \caption{A $6$-wheel of $8$-cuts with a center of size $|T|=2$.}
    \label{fig:wheel}
\end{figure}
Note that a $w$-wheel $(T; C_1, C_2, \ldots, C_w)$
contains $x$-wheels, $x\in[4,w-1]$.  Specifically, for \emph{any} 
subset $\{i_1, i_2, \ldots, i_x\}\subseteq \{1, 2, \ldots, w\}$ 
with $x\geq 4$, 
$(T; C_{i_1}, C_{i_2}, \ldots, C_{i_x})$ forms an $x$-wheel
called a \emph{subwheel} of the original.
If a wheel is not a subwheel of any other wheel, 
it is a \emph{maximal wheel}.
If there exists an index $i^\sharp$ such that, $\sum_{i\neq i^\sharp}\abs{S_i}\leq \kappa$, 
then we say this is a \emph{small wheel}.\footnote{For a small 
wheel, all its cuts $C(i,j)$ are
(\Rmnum{2}, $O(\kappa^2)$)-small.}

\paragraph{Matching Cuts and Crossing Matching Cuts.}
Let $U$ be a cut, $A$ be a side of $U$, and $P\subseteq U$ be a subset of the cut.
We call a cut $W$ a \emph{matching cut of $U$ in side $A$ w.r.t.~$P$} if
(i) $U\backslash P \subseteq W\subseteq U\cup A$,
(ii) $A\backslash W\neq\emptyset$, and 
(iii) $W$ disconnects $P\cup(V\backslash(U\cup A))$ from $A\backslash W$.
The set $\Match{U;A}(P) \bydef W\backslash U$ is the neighborhood of $P$ 
restricted to $A$.  Note that a matching cut is a type of laminar cut.

Now suppose $U$ is a cut with exactly two sides $A$ and $B$, and let $P\subseteq U$ be a non-empty subset of $U$.
We call $W$ 
a \emph{crossing matching cut of $U$ in
side $A$ w.r.t.~$P$}
if
(i) $W\cap B\neq \emptyset$,
(ii) $(U\backslash P)\cup(W\cap A)$ is a matching cut of $U$ in side $A$ w.r.t.~$P$, 
\begin{figure}[h!]
    \centering
    \scalebox{.55}{\includegraphics{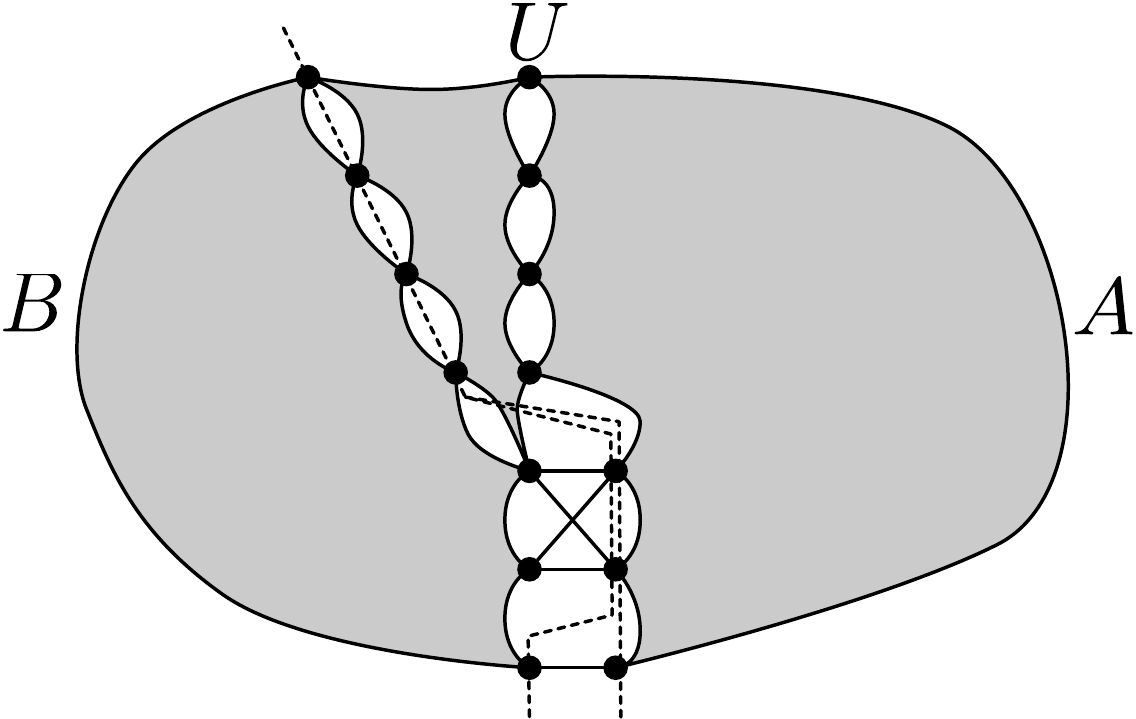}}
    \caption{A cut $U$ (drawn vertically) with two sides $A$ and $B$.  Dotted
    lines indicate two crossing matching cuts w.r.t.~$P_1$ (bottom 3 vertices of $U$)
    and $P_2$ (top 2 vertices of $P_1$).}
    \label{fig:crossing-matching-cut}
\end{figure}

One could view $U$ and a crossing matching cut $W$ as a degenerate $4$-wheel, 
in which one sector $S_1=\emptyset$ is empty.  
Such cuts should \emph{not} be regarded as wheels, as they do not possess 
key properties of wheels,
e.g., that when $U$ and $W$ are (minimum) $\kappa$-cuts, that 
$|C_1|=\cdots=|C_4|=\frac{\kappa-\abs{T}}{2}$, because $C_1\cup T\cup C_2$ is not a cut.

\bigbreak

Lemmas~\ref{lem1} and \ref{lem2} are used throughout the paper.
Recall here $\kappa=\kappa(G)$ is the vertex connectivity of $G$.

\begin{lemma}\label{lem1}
  Suppose $U$ is a $\kappa$-cut
  and $P$ a side of $U$. 
  For every $p\in P$ and $u\in U$, 
  there exists a path from $p$ to 
  $u$ that is not blocked by $V\backslash P$.
\end{lemma}
\begin{proof}
Fix any $v$ in another side of $U$. 
By Menger's thorem (Theorem~\ref{thm:Menger})
there are $\kappa$ internally vertex disjoint paths
from $u$ to $v$, and therefore each must pass through a
different vertex of $U$.  
The prefixes of these paths that are contained in 
$P\cup U$ are not blocked by $V\backslash P$.
\end{proof}

\begin{lemma}\label{lem2}
  Suppose $U$ and $W$ are two cuts, $P$ is disconnected by $U$ from the rest of the graph $G$ and $Q$ is disconnected by $W$ from the rest of the graph $G$. Then we have the following two rules:
  \begin{itemize}
    \item (Intersection Rule) If $P\cap Q\neq \emptyset$, then $P\cap Q$ is disconnected by $(U\cap Q) \cup (U\cap W) \cup (W\cap P)$ from the rest of the graph $G$;
    \item (Union Rule) If $V\backslash(U\cup P\cup W\cup Q)\neq \emptyset$, then $P\cup Q$ is disconnected by $(U\backslash Q)\cup (W\backslash P)$ from the rest of the graph $G$.
  \end{itemize}
\end{lemma}
\begin{figure}[h!]
    \centering
    \begin{tabular}{c@{\hspace*{2cm}}c}
    \scalebox{.35}{\includegraphics{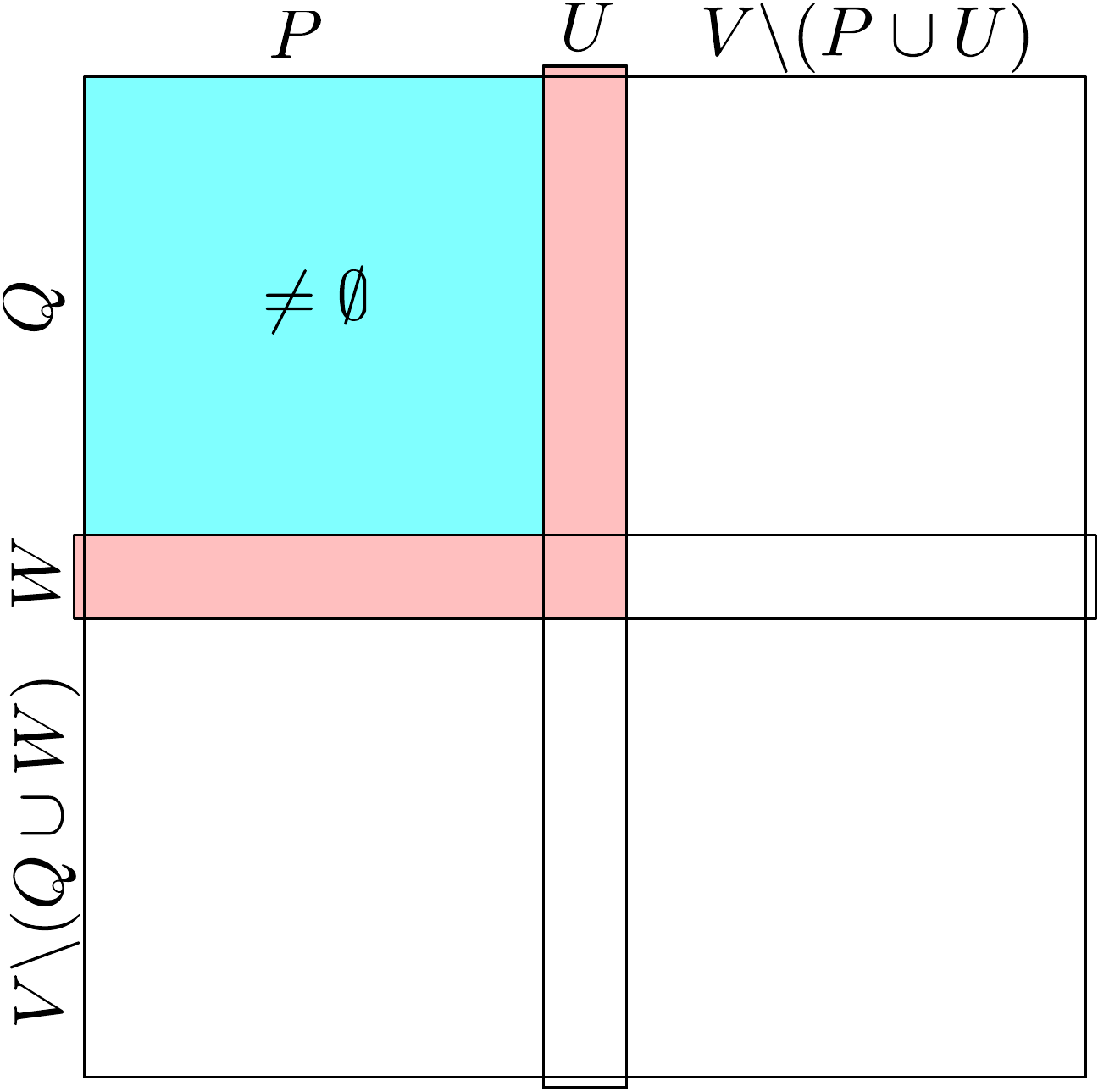}}
    &
    \scalebox{.35}{\includegraphics{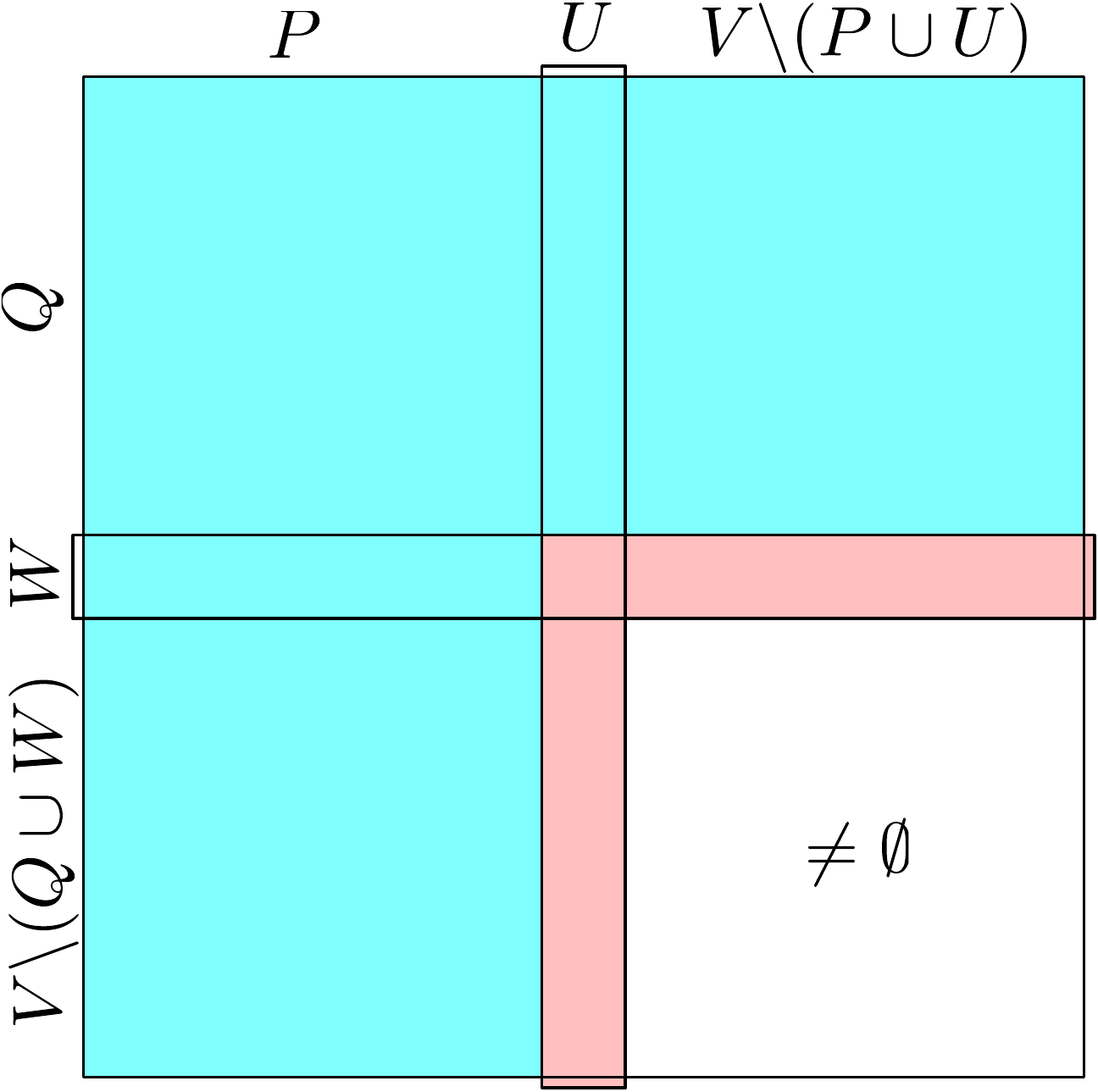}}\\
 (a) & (b)
    \end{tabular}
    \caption{(a) Intersection rule; (b) Union rule.}
    \label{fig:union-intersection}
\end{figure}
\begin{proof}
Denote $A=(U\cap Q)\cup(U\cap W)\cup (W\cap P)$, $B=P\cap Q$. First, $A\cup B = (U\cup P)\cap (W\cup Q)\subset V$, so $V\backslash(A\cup B)\neq \emptyset$. Consider a path $v_1v_2\cdots v_l$ from $B$ to $V\backslash(A\cup B)$. Because $v_1\in B = P\cap Q$, and $v_l\notin U\cup P$ or $v_l\notin W\cup Q$, this path must be blocked by $U$ or $W$, so there must exist some $v_j$ in $U\cup W$. Find the smallest $i$ such that $v_{i}\in U\cup W$, 
and without loss of generality assume $v_i\in U$. 
Then the path $v_1v_2\cdots v_i$ is not blocked by $W$, so $v_i\in Q\cup W$.
Since $(Q\cup W)\cap U\subseteq A$ we have $v_i\in A$.
It follows that $A$ separates $B$ from the rest of the graph,
proving the Intersection Rule.

For the Union Rule, let $R=V\backslash(U\cup P)$, $S=V\backslash(W\cup Q)$. Now that $R\cap S=V\backslash(U\cup P\cup W\cup Q)\neq \emptyset$, 
by applying the intersection rule above, $R\cap S$ is disconnected by 
$T=(U\cap S)\cup (U\cap W)\cup (W\cap R) 
  =(U\backslash Q)\cup (W\backslash P)$ 
from the rest of the graph.
\end{proof}

\section{The Classification of Minimum Vertex Cuts}\label{sect:classification}

The main \emph{binary} structural theorem for 
vertex connectivity is, informally, that every two minimum vertex
cuts have a relationship that is \emph{Laminar}, \emph{Wheel}, 
\emph{Crossing Matching}, or \emph{Small}; cf.~\cite{cohen1993reinventing}.  
Moreover, any strict subset of this
list would be inadequate to capture all possible relationships 
between two vertex cuts.\footnote{The existence of \emph{Small} 
cuts as a category---an \emph{a priori}
unnatural class---indicates that there may be other ways to capture
all minimum vertex cuts through an entirely different classification system.}

\begin{theorem}
\label{thm1}
    Fix a minimum $\kappa$-cut $U$
    with sides $A_1, A_2, \ldots, A_a$, $a\geq 2$, 
    and let $W$ be any other $\kappa$-cut 
    with sides $B_1, B_2, \ldots, B_b$, $b\geq 2$. 
    Denote $T=U\cap W$, $W_i=W\cap A_i$ and $U_j=U\cap B_j$. 
    Then $W$ may be classified w.r.t.~$U$ as follows:
    \begin{description}
      \item [Laminar type.] $W$ is a laminar cut of $U$, and in particular, there exists indices $i^*$ and $j^*$ such that $B_{j^*}\backslash A_{i^*}=(U\backslash W)\cup (\cup_{i\neq i^*}A_i)$ and $A_{i^*}\backslash B_{j^*}=(W\backslash U)\cup (\cup_{j\neq j^*}B_{j})$.
      \item [Wheel type.] $a=b=2$, and $(T; U_1, W_1, U_2, W_2)$ forms a $4$-wheel with sectors $A_1\cap B_1$, $A_1\cap B_2$, $A_2\cap B_2$ and $A_2\cap B_1$.
      \item [Crossing Matching type.] $a=b=2$, and w.l.o.g.,
      $A_1\cap B_1\neq \emptyset$, $A_2\cap B_2\neq\emptyset$,
      but $A_1\cap B_2=\emptyset$.  We have $\abs{W_2}=\abs{U_1}>0$,
        $\abs{W_1}=\abs{U_2}>0$, and $W$ is a crossing matching 
        cut of $U$ in side $A_1$ w.r.t.~$U_2$.  Furthermore, 
        if $A_2\cap B_1\neq \emptyset$, then $\abs{U_1}\geq \abs{U_2}$.

      \item [Small type.] $U$ is (\Rmnum{1}, $\kappa-1$)-small, 
      and the small sides of $U$ are within $W$,
      or $W$ is (\Rmnum{1}, $\kappa-1$)-small, and 
      the small sides of $W$ are within $U$.
    \end{description}
\end{theorem}

\begin{proof}
Suppose there is a single index $i^*$ such that 
$W_{i^*}\neq \emptyset$ and $W_i=\emptyset$ for all $i\neq i^*$.
It follows that $W\subseteq A_{i^*}\cup U$ is a laminar cut
of $U$ in side $A_{i^*}$.  It remains to prove the other properties
of the laminar type.
By Lemma~\ref{lem1} there exists paths from any vertex in 
$A_i$, $i\neq i^*$,
to $U\backslash W$ that are not blocked by $W$, 
so they all lie within one side of $W$;
let us denote this side by $B_{j^*}$. 
Then $(U\backslash W)\cup(\cup_{i\neq i^*}A_i)\subseteq B_{j^*}$,
and because $V=U\cup (\cup_{i=1}^a A_i)$, we obtain 
$B_{j^*}\backslash A_{i^*}=(U\backslash W)\cup (\cup_{i\neq i^*} A_i)$. 
Now that $U\subseteq W\cup B_{j^*}$ is laminar w.r.t.~$W$, 
so based on the same reasoning 
we have $A_{i^*}\backslash B_{j^*}=(W\backslash U)\cup(\cup_{j\neq j^*}B_j)$.

\medskip

We proceed under the assumption that such 
indices $i^*,j^*$ do not exist,
and without loss of generality assume 
that $W_1, W_2, U_1, U_2\neq \emptyset$.
We now wish to prove that \emph{all} $U_i,W_i$ are non-empty.
Suppose $W_i \bydef W\cap A_i = \emptyset$ were empty, 
then $A_i$ would be contained 
within a side of $W$, say $A_i\subseteq B_j$.
By Lemma~\ref{lem2} (intersection rule), 
whenever $A_i\cap B_j\neq \emptyset$, 
the set $W_i\cup T\cup U_j$ disconnects 
$A_i\cap B_j$ from the rest of the graph.
It follows that
\[
\abs{W_i}+\abs{T}+\abs{U_j}=\abs{T}+\abs{U_j}\geq \kappa =\abs{T}+\sum_{l=1}^{b}\abs{U_l},
\]
which implies that $U_j$ is the \emph{only} 
non-empty $U_*$-set, 
contradicting $U_1, U_2\neq \emptyset$.
Therefore, $W_i\neq \emptyset$ for all $i$ and similarly, 
$U_j\neq \emptyset$ for all $j$.

\medskip 

Define $\Omega=\{(i, j) \mid A_i\cap B_j\neq \emptyset\}$ to be the side-pairs whose intersections are non-empty. 
We consider the following possibilities, which are exhaustive.
\begin{itemize}
  \item[$1$\cir] There exist $(i, j), (i', j')\in \Omega$ such 
  that $i\neq i'$, $j\neq j'$. Then by Lemma~\ref{lem2} (intersection rule)
\begin{align*}
\abs{W_i}+\abs{T}+\abs{U_j}                       &\geq \kappa\\ 
\mbox{and \ } \abs{W_{i'}}+\abs{T}+\abs{U_{j'}}   &\geq \kappa.
\intertext{On the other hand,}
\abs{U_{j}}+\abs{U_{j'}}+\abs{T} &\leq \abs{U} = \kappa\\ \mbox{ and \ } \abs{W_{i}}+\abs{W_{i'}}+\abs{T} &\leq \abs{W} = \kappa.
\end{align*}
Thus all these inequalities must be equalities, 
and, adding the fact that all $W_i, U_j\neq \emptyset$, 
we conclude that $a=b=2$, 
$\abs{W_i}=\abs{U_{j'}}$, $\abs{W_{i'}}=\abs{U_j}$.
W.l.o.g.~we fix $i=j=1$, $i'=j'=2$.
See Figure~\ref{fig:Thm2-case1}.
\begin{figure}[h!]
    \centering
    \scalebox{.35}{\includegraphics{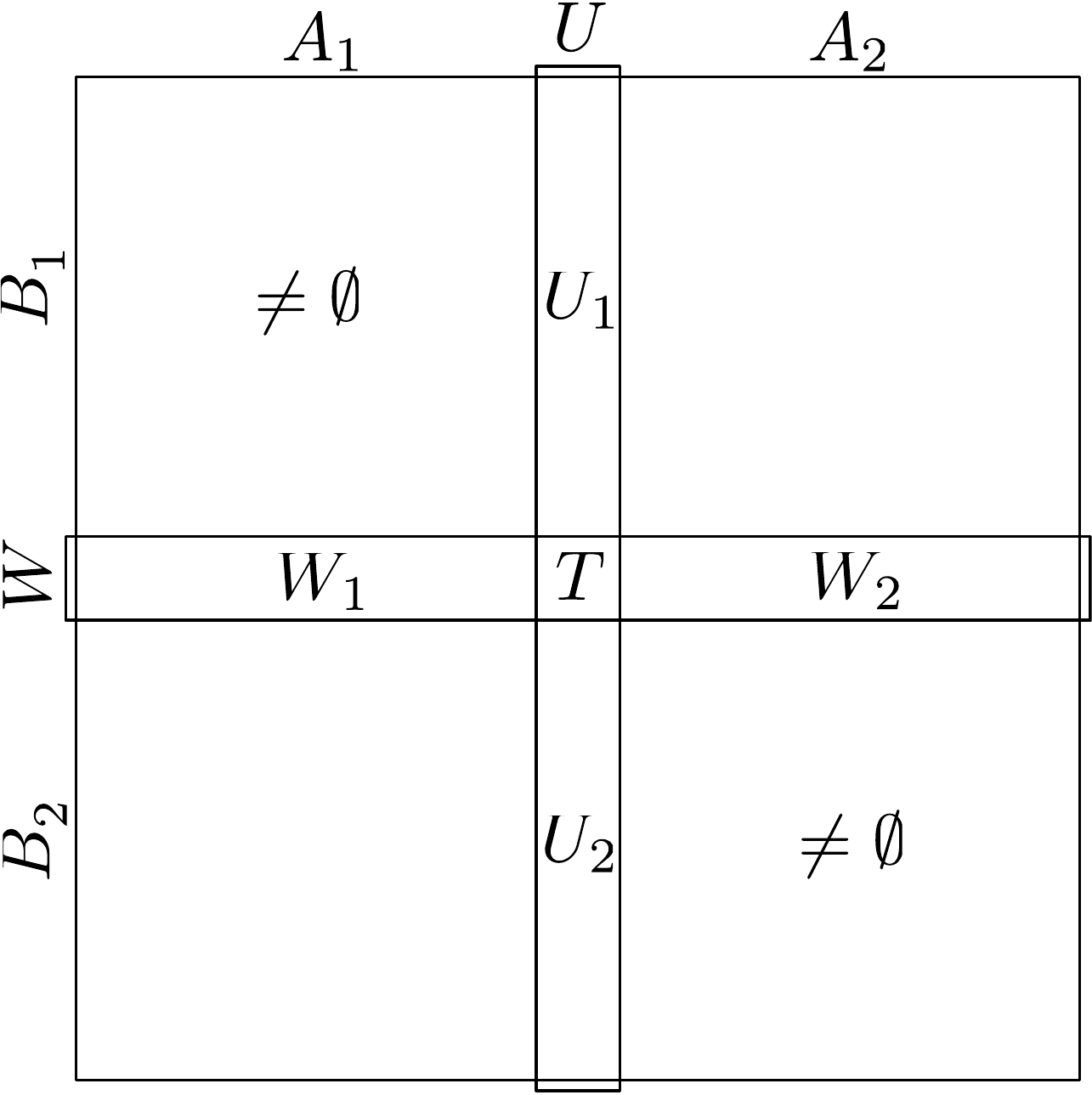}}
    \caption{A depiction of cuts $U,W$ in case 1\cir.}
    \label{fig:Thm2-case1}
\end{figure}
      \begin{itemize}
        \item[1.1\cir] Suppose $A_1\cap B_2\neq \emptyset$ and $A_2\cap B_1\neq \emptyset$. Then $\abs{W_{\hat{\imath}}}+\abs{U_{\hat{\jmath}}}\geq \kappa - \abs{T}$ 
        for every $\hat{\imath}, \hat{\jmath}\in \{1, 2\}$, 
        so we conclude that 
        \[
        \abs{U_1}=\abs{U_2}=\abs{W_1}=\abs{W_2}=\frac{\kappa-\abs{T}}{2}.
        \]
        Now that $W_{\hat{\imath}}\cup T\cup U_{\hat{\jmath}}$ disconnects $A_{\hat{\imath}}\cap B_{\hat{\jmath}}$ from the rest of the graph, $U_1\cup T\cup U_2=U$ disconnects $(A_1\cap B_1)\cup W_1\cup(A_1\cap B_2)=A_1$ from $(A_2\cap B_1)\cup W_2\cup(A_2\cap B_2)=A_2$, $W_1\cup T\cup W_2=W$ disconnects $(A_1\cap B_1)\cup U_1\cup(A_2\cap B_1)=B_1$ from $(A_1\cap B_2)\cup U_2\cup(A_2\cap B_2)=B_2$, we conclude that $(T; U_1, W_1, U_2, W_2)$ forms a $4$-wheel.
        \begin{figure}[h!]
            \centering
            \scalebox{.4}{\includegraphics{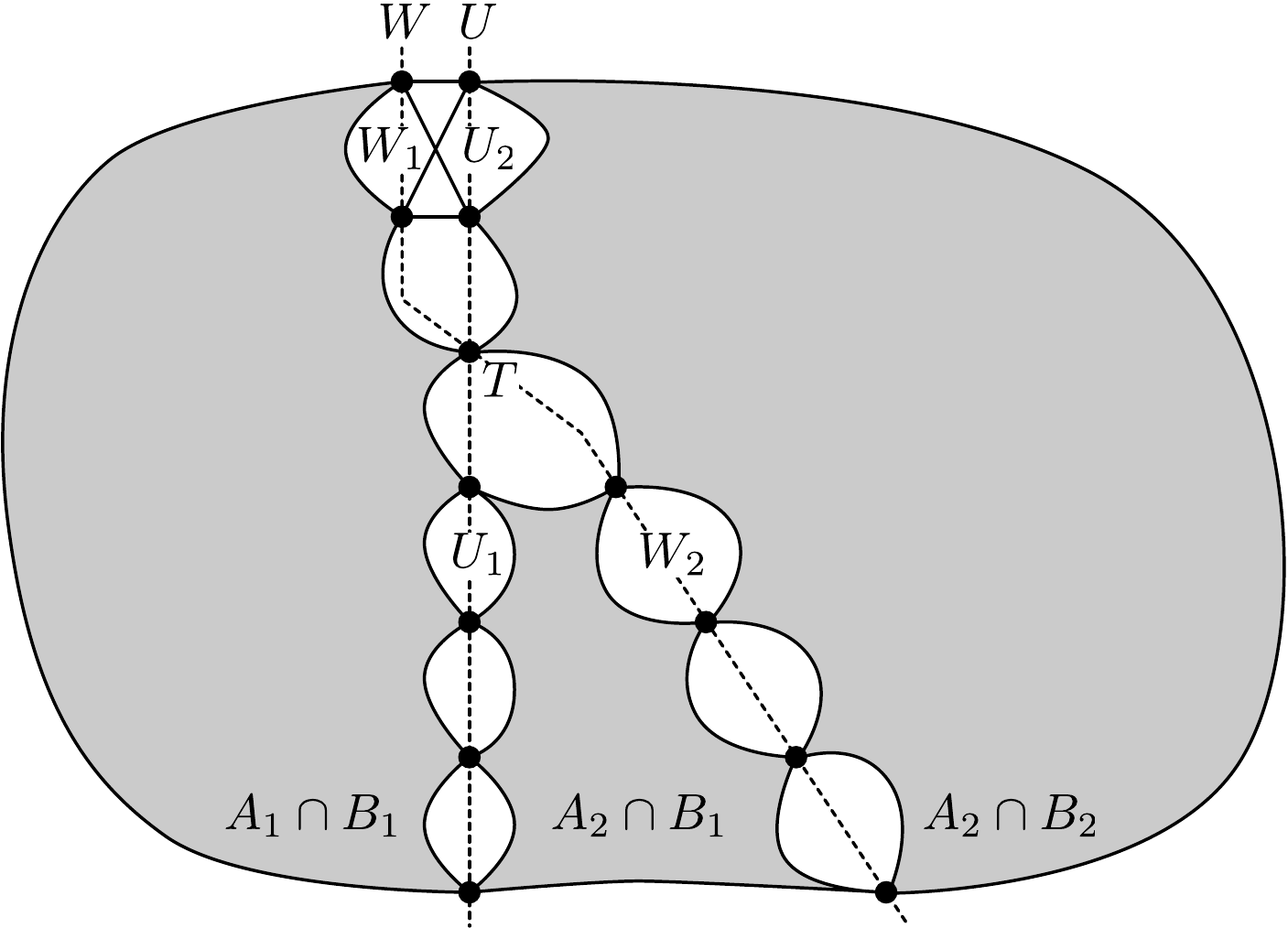}}
            \caption{A depiction of the cuts $U,W$ in case 1.2\cir.}
            \label{fig:Thm1-crossing-matching}
        \end{figure}
        \item[1.2\cir] Suppose
        $A_1\cap B_{2}=\emptyset$ (or symmetrically, that 
        $A_2\cap B_1=\emptyset$).
        Then $A_1=A_1\cap (B_1\cup W)=(A_1\cap B_1)\cup W_1$. By Lemma~\ref{lem2}, $W_1\cup T\cup U_1$
        separates $A_1\cap B_1$ from the rest of the 
        graph.  Since 
        $U_{2}\subseteq V\backslash((A_1\cap B_1)\cup(U_1\cup T\cup W_1))$, it follows that
        $W_1\cup T\cup U_1$ disconnects 
        $U_2$ from
        $A_1\cap B_1=A_1\backslash(W_1\cup T\cup U_1)$,
        i.e., it is a matching cut of $U$ 
        in side $A_1$ w.r.t.~$U_2$.
        See Figure~\ref{fig:Thm1-crossing-matching}.
        Because $W_2=W\cap A_2\neq \emptyset$ and $(W\cap A_1)\cup (U\backslash U_{2})=W_1\cup T\cup U_1$, 
        $W$ is a crossing matching cut of $U$ in 
        side $A_1$ w.r.t.~$U_2$.
        
        If $A_2\cap B_1\neq \emptyset$, by Lemma~\ref{lem2} (intersection rule)
        \[\abs{U_1}+\abs{W_2}+\abs{T}\geq \kappa = \abs{W_1}+\abs{T}+\abs{W_2},\]
        so $\abs{U_1}\geq \abs{U_2}$.
      \end{itemize}
      
  \item [2\cir] Suppose there exists a $j^\sharp$
  such that $\forall i.\forall j\neq j^\sharp. (i,j)\not\in \Omega$,
  i.e., $A_i\cap B_{j} = \emptyset$. 
  This implies that $\cup_{j\neq j^\sharp} B_j\subseteq U$, 
  and because $U_{j^{\sharp}}\neq \emptyset$, 
  $\abs{\cup_{j\neq j^\sharp}B_j}$ is \emph{strictly} smaller than $\kappa$. 
  Therefore $W$ is a (\Rmnum{1}, $\kappa-1$)-small cut, 
  and all the small sides of $W$ are within $U$.
  
  \item [3\cir] There exists $i^\sharp$ such that 
  $\forall i\neq i^\sharp.\forall j. (i,j)\not\in \Omega$. 
  Symmetric to case 2\cir; $U$ is (\Rmnum{1}, $\kappa-1$)-small, 
  and all the small sides of $U$ are within $W$.

  \item [4\cir] $\Omega = \emptyset$. Then $\cup_{i=1}^a A_i\subseteq W$, 
  so $V=U\cup (\cup_{i=1}^a A_i)\subseteq U\cup W$, 
  and $\abs{V}\leq 2\kappa$.  This is a possibility, 
  but not one we consider as it contradicts our initial assumption that $n>4\kappa$.
\end{itemize}
\end{proof}

\begin{corollary}
  \label{cor1}
  If $U$ is a $\kappa$-cut that is not (\Rmnum{1}, $\kappa-1$)-small 
  and has at least $3$ sides, 
  then all other $\kappa$-cuts have a laminar type relation with $U$, 
  or are themselves (\Rmnum{1}, $\kappa-1$)-small cuts.
\end{corollary}

\begin{corollary}
  \label{cor2}
  Suppose $U$ is a $\kappa$-cut that is not (\Rmnum{1}, $\kappa-1$)-small, 
  with exactly two sides $A$ and $B$.
  Suppose $W$ is a $\kappa$-cut with sides $K$, $L$ (and possibly others), 
  such that $W\cap A\neq \emptyset$, $W\cap B\neq \emptyset$, $A\subseteq K\cup W$, 
  and $L\cap U\neq \emptyset$.
  Then $W$ only has two sides, and $W$ is a crossing matching cut of $U$
  in side $A$ w.r.t.~$L\cap U$.
\end{corollary}

\begin{corollary}\label{cor:CD}
Define $\Cuts{C; D}$ to be the 
set of all $\kappa$-cuts 
that disconnect disjoint,
non-empty vertex sets $C$ and $D$. 
If $\Cuts{C; D}\neq \emptyset$, 
it contains a unique minimal element $\MinCut{C;D}$, 
such that for any cut $U\in \Cuts{C; D}$, $\Region{\MinCut{C; D}}(C) \subseteq \Region{U}(C)$.
\end{corollary}

\begin{proof}
This is a corollary of Theorem~\ref{thm1}, but also admits a simple, direct proof via the Picard-Queyrenne theorem~\cite{PicardQ80}.  Form a flow network $\vec{G}$ via the following steps (i) contract $C$ and $D$ to vertices $s$ and $t$, (ii) replace each vertex $v$ with a subgraph consisting of vertices $v_{\operatorname{in}},v_{\operatorname{out}}$
and a directed edge $(v_{\operatorname{in}},v_{\operatorname{out}})$,
(iii) replace each undirected edge $\{u,v\}$ with
directed edges 
$(u_{\operatorname{out}},v_{\operatorname{in}})$
and 
$(v_{\operatorname{out}},u_{\operatorname{in}})$,
(iv) give edges from (ii) unit capacity
and edges from (iii) infinite capacity.  
If the flow value is $\kappa$, then 
the minimum $(s_{\operatorname{out}},t_{\operatorname{in}})$-cuts
are in one-to-one correspondence with
the minimum vertex cuts in $\Cuts{C;D}\neq\emptyset$.
Since $(s_{\operatorname{out}},t_{\operatorname{in}})$-cuts in $\vec{G}$
are closed under union and intersection~\cite{PicardQ80},
there is a unique vertex cut $\MinCut{C;D} \in \Cuts{C;D}$
such that $\Region{\MinCut{C;D}}(C)$ is minimal w.r.t.~containment.

\end{proof}

Theorem~\ref{thm1} classifies the pairwise relationship between two 
minimum $\kappa$-cuts.  
In 
Sections~\ref{subsect:wheel}--\ref{subsect:small} 
we further explore the properties of 
wheel cuts, (crossing) matching cuts,
laminar cuts, and small cuts.

\subsection{Wheels and Wheel Cuts}\label{subsect:wheel}

Recall that a $w$-wheel $(T; C_1,\ldots,C_w)$ satisfied, by definition,
the property that $C_i\cup T\cup C_{i+2}$ formed a $\kappa$-cut, but
did not say anything explicitly about $C(i,j) = C_i \cup T \cup C_j$.
Lemma~\ref{lem3} proves that these are also cuts, and bounds their number
of sides.

\begin{lemma}
  \label{lem3}
  Suppose $(T; C_1, C_2, \ldots, C_w)$ forms a $w$-wheel ($w\geq 4$) with 
  sectors $S_1, S_2, \ldots, S_w$. (Subscripts are modulo $w$.) 
  For any $i\neq j$, $C(i, j)$ is a $\kappa$-cut 
  that disconnects $D(i, j)$ from the rest of the graph.
  Moreover, when $j-i \not\in \{1,w-1\}$, $C(i, j)$ has exactly two sides, 
  which are $D(i, j)$ and $D(j, i)$.
  Furthermore, $\abs{C_i}=\frac{\kappa-\abs{T}}{2}$.
\end{lemma}

\begin{proof}
By definition $C_i\cup T\cup C_{i+2}$ and $C_{i-1}\cup T\cup C_{i+1}$ are two 
$\kappa$-cuts that, respectively, 
separate $S_i\cup C_{i+1}\cup S_{i+1}$ and $S_{i-1}\cup C_{i}\cup S_{i}$ 
from the rest of the graph.
By Lemma~\ref{lem2} (intersection rule), $C_i\cup T\cup C_{i+1}$ disconnects $S_i$ from the rest of the graph.
Thus, whenever $j-i \in \{1,2\}$, $C(i, j)$ disconnects $D(i, j)$ from the rest of the graph. 
This is the base case. Assuming the claim is true whenever $j-i \in [1,l-1]$, we prove
it is true up to $l$ as well, $l\leq w-1$.
Fix $i,j$ such that $j-i=l$.  Then $C(i, j-1)$ is a cut that disconnects $D(i, j-1)$ from the rest of the graph, and $C(i+1, j)$ is a cut that disconnects $D(i+1, j)$ from the rest of the graph. 
By Lemma~\ref{lem2} (union rule), $(C(i, j-1)\backslash D(i+1, j)\cup (C(i+1, j)\backslash D(i, j-1))=C_i\cup T\cup C_j$ disconnects $D(i, j-1)\cup D(i+1, j)=D(i, j)$ from the rest of the graph. 
This proves the first part.

By Lemma~\ref{lem1}, 
there exist paths from any vertex in $S_r$ to every vertex in $C_r$, 
and to every vertex in $C_{r+1}$, 
that is not blocked by $V\backslash S_r$. 
Thus, when $j-i \not\in \{1,w-1\}$, $C(i,j)$ separates $D(i,j)$ from $D(j,i)$,
$D(i,j)$ forms a side since all vertices in $D(i,j)$ 
have paths to $C_{i+1}$, being distinct from $C_j$, 
and $D(j,i)$ forms a side since all vertices have paths 
to $C_{j+1}$, being distinct from $C_i$.
(When $j=i+1$, $D(i,i+1)$ may be a region consisting of multiple sides.)

Now it is proved that for all $i, j$, $C(i, j)$ is a $\kappa$-cut, so $\abs{C_i}+\abs{C_j}=\kappa - \abs{T}$ for all $i, j$. Therefore all $\abs{C_i}$ are equal to $\frac{\kappa - \abs{T}}{2}$.
\end{proof}

\begin{remark}
Lemma~\ref{lem3} shows that the set $\{C(i,i+2)\}_{i\in [w]}$
generates all the ${w\choose 2}$ wheel cuts.  One might
think that the sector cuts $\{C(i,i+1)\}$ would also suffice,
but this is incorrect.  In Figure~\ref{fig:faux-wheel},
$C(i,i+1)$ is a (minimum) 6-cut for all $i$ separating $S_i$
from the rest of the graph, but this is \emph{not} a 4-wheel 
since $C(1,3)$ and $C(2,4)$ \emph{are not cuts}.
\begin{figure}[h!]
    \centering
    \scalebox{.4}{\includegraphics{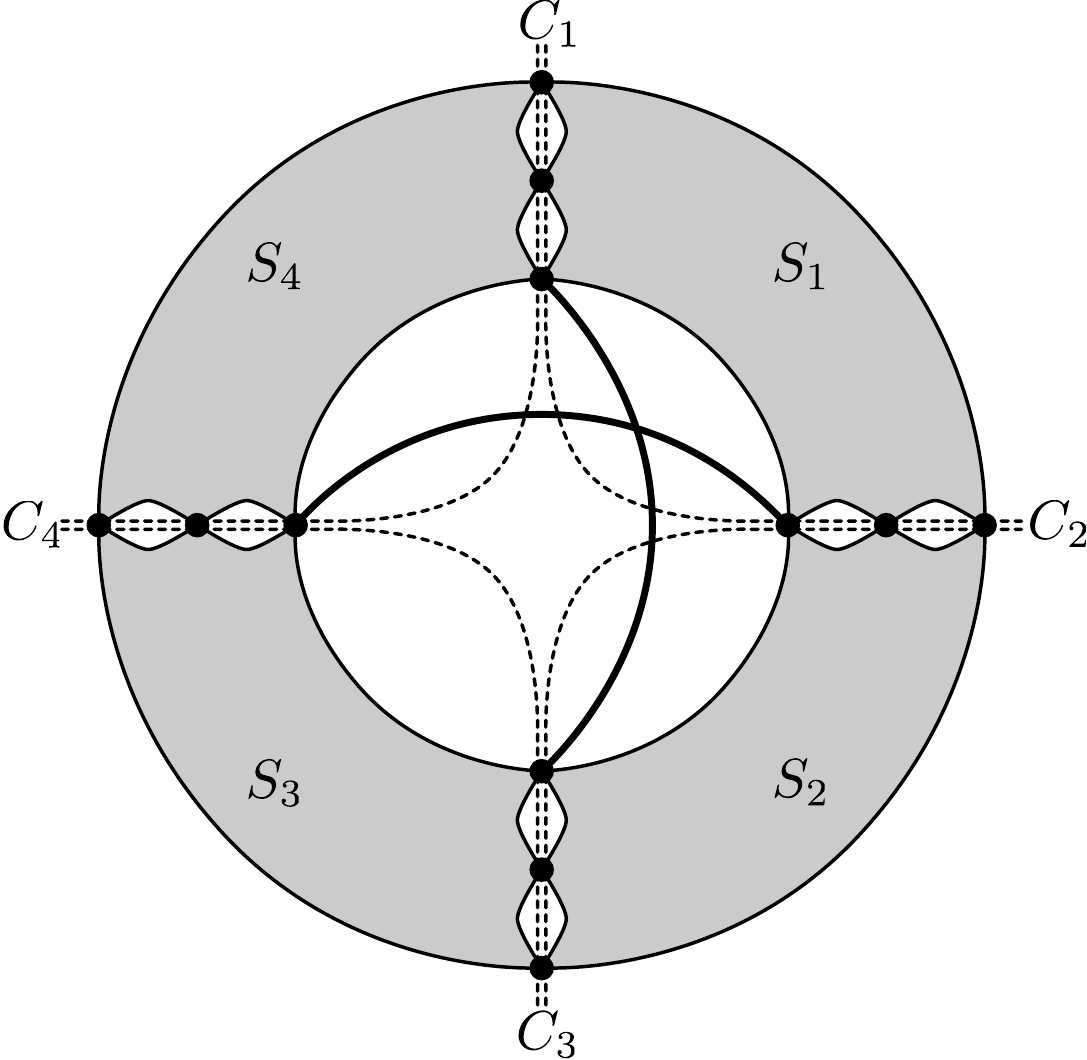}}
    \caption{A faux 4-wheel.}
    \label{fig:faux-wheel}
\end{figure}
\end{remark}

\begin{theorem}
\label{thm2}
   Suppose $(T; C_1, C_2, \ldots, C_w)$ forms a $w$-wheel ($w\geq 4$) 
   with sectors $S_1, S_2, \ldots, S_w$. (Subscripts are given by modulo $w$.) 
   Let $X$ be any minimum $\kappa$-cut. Then one of the following is true:
  \begin{itemize}
    \item[$1$\cir] $X=C(i, j)$ for some $i\neq j$.
    \item[$2$\cir] $X\subseteq C(i, i+1)\cup S_i$ for some $i$, 
    i.e., $X$ is a laminar cut of $C(i, i+1)$.
    \item[$3$\cir] $X$ has crossing matching type relation with some $C(i, i+1)$ or some $C(i, i+2)$.
    \item[$4$\cir] $X$ is a (\Rmnum{1}, $\kappa-1$)-small cut.
    \item[$5$\cir] $(T; C_1, C_2, \ldots, C_w)$ is a small wheel.
    \item[$6$\cir] There exists $i<j$, such that $X\subseteq S_i\cup T\cup S_j$, and $(T; C_1, \ldots, C_i, X\cap S_i, C_{i+1}, \ldots, C_j, X\cap S_j, C_{j+1}, \ldots, C_w)$ 
    forms a $(w+2)$-wheel; or there exists $i\neq j$, $X\subseteq S_i\cup T\cup C_j$, and $(T; C_1, \ldots, C_i, X\cap S_i, C_{i+1}, \ldots, C_w)$ forms a $(w+1)$-wheel. In other words, $(T; C_1, C_2, \ldots, C_w)$ is a subwheel of some other wheel.
  \end{itemize}
\end{theorem}

\begin{proof}
By the definition of wheels and Lemma~\ref{lem1}, there exists paths from any vertex in 
$S_i$ to $C_i$, to $T$, and to $C_{i+1}$ that are not blocked by $V\backslash S_i$,
and there exist paths from any vertex in $C_i$ to $T$ that are not blocked by 
$V\backslash (S_{i-1}\cup C_i\cup S_i)$. These facts are used frequently below.

\medskip

For every $i\in\{1, 2, \ldots, w\}$, we consider the 
following cases relating cuts $X$ and $C(i, i+1)$,
which are exhaustive according to Theorem~\ref{thm1}.
\begin{itemize}
\item[\rmnum{1}.] $X=C(i, i+1)$, this is a trivial case and we are in case $1$\cir. 
\item[\rmnum{2}.] $X\cap S_i=\emptyset$.
\item[\rmnum{3}.] $X$ is a laminar cut of $C(i, i+1)$, and $X\cap S_i\neq \emptyset$, then it must be true that $X\subseteq C(i, i+1)\cup S_i$, and we are in case $2$\cir.
\item[\rmnum{4}.] $X$ is a wheel type cut of $C(i, i+1)$, and $X\cap S_i\neq \emptyset$.
\item[\rmnum{5}.] $X$ is a crossing matching type cut of $C(i, i+1)$, then we are in case $3$\cir.
\item[\rmnum{6}.] $X$ and $C(i,i+1)$ have a small type relation, such that $X$ is (\Rmnum{1},$\kappa-1$)-small, then we are in case $4$\cir. 
\item[\rmnum{7}.] $X$ and $C(i, i+1)$ have a small type relation, such that $C(i, i+1)$ is (\Rmnum{1}, $\kappa-1$)-small, and $D(i+1, i)$ is the small side, which implies that $\abs{D(i, i+1)}<\kappa$, so the wheel is a small wheel, and we are in case $5$\cir.
\item[\rmnum{8}.] $X$ and $C(i, i+1)$ have a small-type relation, such that $C(i, i+1)$ is (\Rmnum{1}, $\kappa-1$)-small, and $S_i=D(i, i+1)$ is the region of the small sides. 
Then we have that $S_i=D(i, i+1)\subseteq X$.
\end{itemize}

It can be seen that if for any $i$, we are in case 
\rmnum{1}, \rmnum{3}, \rmnum{5}, \rmnum{6}, or \rmnum{7}, 
there is nothing left to prove. 
Otherwise, we may proceed under the assumption 
that \emph{every} index $i$ is in 
case \rmnum{2}, \rmnum{4}, or \rmnum{8}.
We define the index sets $I_1, I_2, I_3$ as follows.
\begin{align*}
I_1 &= \mbox{$\{i \mid $ case \rmnum{4}\ applies to $i\}$},\\
I_2 &= \mbox{$\{i \mid $ case \rmnum{8}\ applies to $i\}$},\\
I_3 &= \mbox{$\{i \mid $ case \rmnum{2}\ applies to $i\}$},\\
\mbox{ and hence \ }
I_1\cup I_2\cup I_3 &= [w] = \{1, 2, \ldots, w\}. 
\end{align*}

We split the possibilities into the following cases.
\begin{itemize}
  \item[I.] $\abs{I_1}\geq 2$.
  \item[II.] $\abs{I_1}\leq 1$, $\abs{I_2}>0$, $\abs{I_3}>0$.
  \item[III.] $\abs{I_1}= 1$, $\abs{I_2}=0$.
  \item[IV.] $\abs{I_1}\leq 1$, $\abs{I_3}=0$.
  \item[V.] $\abs{I_1}=\abs{I_2}=0$.
\end{itemize}

We will show that I and III lead to case $6$\cir, II leads to case $3$\cir, IV leads to case $5$\cir, and V leads to $1$\cir.  Define $Y_i,Z_i,Q$ as follows
\begin{align*}
            Y_i &= X\cap S_i,\\ 
            Z_i &= X\cap C_i,\\
\mbox{and } Q &= X\cap T.
\end{align*}

\begin{itemize}
\item[I.] Suppose $i\neq j$ are in $I_1$. 
Because $X$ has a wheel type relation with some other cut (namely $C(i, i+1)$ and $C(j, j+1)$), it must have exactly two sides; 
let them be $K$ and $L$.
Without loss of generality we assume $i=1$. 
We consider two subcases depending on whether
$j\not\in \{2,w\}$ (I.a) or $j\in \{2,w\}$ (I.b).
\begin{itemize}
\item [I.a.] $j\not\in \{2,w\}$, i.e., 
$j$ is not adjacent to $i$.
We prove the following claims, 
culminating in Claim 3, which puts us
in case 6\cir.
\begin{description}
  \item [Claim 1.] $X\subseteq S_1\cup T\cup S_j$.

  Because $X$ has a wheel type relation 
  with $C(1, 2)$ and $C(j, j+1)$, 
  $\abs{X\cap D(1,2)}=\abs{X\cap D(2,1)}$
  and 
  $\abs{X\cap D(j,j+1)}=\abs{X\cap D(j+1,j)}$.
  Thus, in terms of the $Y_*,Z_*$ sets,
  \[\abs{Y_1}=\sum_{r\not\in\{1, 2\}}\abs{Z_r}+\sum_{r\neq 1}\abs{Y_r}
  \text{ \ \ and \ \ }\abs{Y_j}=\sum_{r\not\in\{j,j+1\}}\abs{Z_r}+\sum_{r\neq j}\abs{Y_r}.\]
  Therefore, $\abs{Z_r}=0$ for all $r$, $\abs{Y_r}=0$ for all $r\not\in\{1, j\}$, and $\abs{Y_1}=\abs{Y_j}=\frac{\kappa-\abs{Q}}{2}$. This means that $I_1=\{1, j\}$, $I_2=\emptyset$,  $I_3= [w]\backslash I_1$, 
  and $X\subseteq S_1\cup T\cup S_j$.
  
  \item [Claim 2.] $C_2\cup D(2, j)\cup C_j$ and $C_{j+1}\cup D(j+1, 1)\cup C_1$ are in different sides of $X$, and $Q=T$.

      By Lemma~\ref{lem1}, the subgraphs induced by 
      $C_2\cup D(2, j)\cup C_j$ and $C_{j+1}\cup D(j+1, 1)\cup C_1$
      are connected, and therefore each is contained in a side 
      of $X$. Suppose they are contained in the \emph{same} side, say $K$.
      Any vertices in $T\backslash X$ must be in $K$ as well, so
      $L\subseteq S_1\cup S_j$ and w.l.o.g.~we assume $L\cap S_1\neq\emptyset$.
      Applying Lemma~\ref{lem2} (intersection rule) to $X$ and $C(1,2)$ we conclude that
      \[
      (L\cap C(1, 2))\cup (S_1\cap X)\cup(X\cap C(1, 2)) 
      = (S_1\cap X)\cup(X\cap C(1, 2)) = 
      Y_1\cup Q
      \]
      is a cut that disconnects $L\cap S_1$ from the rest of the graph.
      It follows that 
      $\abs{Y_1}+\abs{Q}=\frac{\kappa+\abs{Q}}{2}\geq \kappa$, 
      but this contradicts with $\abs{Q}<\kappa$.
      Therefore, we conclude that 
      $C_2\cup D(2, j)\cup C_j$ and 
      $C_{j+1}\cup D(j+1, 1)\cup C_1$ are 
      in \emph{different} sides of $X$,
      and since each of these sides 
      are adjacent to all vertices of $T$,
      that $Q=T$ as well.

  \item [Claim 3.] $(T;C_1, Y_1, C_2, \ldots, C_j, Y_j, C_{j+1},\ldots, C_w)$ forms a $(w+2)$-wheel.
        
      By Claim 2, $T\subseteq X$, so $\abs{Y_1}=\abs{Y_j}=\frac{\kappa-\abs{T}}{2}$ 
      and $C_2\cup D(2, j)\cup C_j \subseteq K$ 
      and $C_{j+1}\cup D(j+1, 1)\cup C_1 \subseteq L$ are in different sides of $X$. 
      To verify that this is a wheel we simply need to confirm that the four new
      cuts involving $Y_i,Y_j$ are in fact cuts. 
      We illustrate this for $Y_1 \cup T \cup C_3$. 
      Applying Lemma~\ref{lem2} (intersection rule) 
      to $X$ with side $K$ and $C(1, 3)$ with side $D(1, 3)$, 
      we conclude that $(K\cap S_1)\cup C_2\cup S_2$ is disconnected 
      by $Y_1\cup T\cup C_3$ from the rest of the graph. 
      The other three new cuts are confirmed similarly, hence 
      $(T;C_1, Y_1, C_2, \ldots, C_j, Y_j, C_{j+1}, \ldots, C_w)$ forms a $(w+2)$-wheel.
\end{description}

\item[I.b.] $j\in\{2,w\}$. W.l.o.g.~we assume $j=2$.
Once again we prove the following claims, culminating in Claim 6 which
puts us in case 6\cir.
\begin{description}
  \item [Claim 4.] $X\subseteq S_1\cup S_2\cup T\cup C_2$.

  Because $X$ has wheel type relation with $C(1, 2)$ and $C(2, 3)$,  we have that
  \[
  \abs{Y_1}=\sum_{r\not\in\{1, 2\}}\abs{Z_r}+\sum_{r\neq 1}\abs{Y_r}
  \text{ \ \ and \ \ }\abs{Y_2}=\sum_{r\not\in\{2, 3\}}\abs{Z_r}+\sum_{r\neq 2}\abs{Y_r}.
  \]
  Therefore, $\abs{Z_r}=0$ for all $r\neq 2$, 
  $\abs{Y_r}=0$ for all $r\not\in\{1, 2\}$, 
  and $\abs{Y_1}=\abs{Y_2}=\frac{\kappa-\abs{Q}-\abs{Z_2}}{2}$. 
  This means $I_1=\{1, 2\}$, $I_2=\emptyset$, 
  $I_3=[w]\backslash I_1$, and $X\subseteq Y_1\cup Y_2 \cup T\cup Z_2$.
  \item [Claim 5.] $Q=T$ and $Z_2=\emptyset$.

  By Lemma~\ref{lem1}, $C_3\cup D(3, w)\cup C_1$ is within a side of $X$, say $K$.
  Moreover, if $T\backslash X\neq \emptyset$, then $T\backslash X\subseteq K$. 
  Because $X$ has a wheel relation with $C(1,2)$,
  \[
  \abs{L\cap C(1, 2)}=\abs{K\cap C(1, 2)}\geq \abs{C_1}.
  \]
  Thus, $X\cap C(1, 2)= Q\cup Z_2$ 
  and
  $\abs{X\cap C(1, 2)}\leq \kappa - 2\abs{C_1}=\abs{T}$.
  Since $T\backslash X\subset K$, 
  it follows that
  \[
  \abs{L\cap C(1, 2)}=\abs{L\cap C_2}\leq \abs{C_2}
  \leq 
  \abs{C_1}+\abs{T\backslash X}=\abs{K\cap C(1, 2)},
  \]
  implying $\abs{T\backslash X}=0$ and therefore that 
  $Q=T$ and $Z_2=\emptyset$.
  \item[Claim 6.]
  $(T;C_1, Y_1, C_2, Y_2, C_3, \ldots, C_w)$ forms a $(w+2)$-wheel.

      By Claims 4 and 5 we know $\abs{Y_1}=\abs{Y_2}=\frac{\kappa-\abs{T}}{2}$. 
      There are three new wheel cuts involving $Y_1,Y_2$
      that need to be confirmed, 
      namely $X=Y_1\cup T\cup Y_2$, which 
      separates $C_3\cup D(3, w)\cup C_1$
      from $C_2$, as well as
      $C_w\cup T\cup Y_1$ and $Y_2\cup T\cup C_4$.
    The latter two are established by applying
    Lemma~\ref{lem2}, as in Claim 3. 
    We conclude that 
    $(T;C_1, Y_1, C_2, Y_2, C_3, \ldots, C_w)$ forms a $(w+2)$-wheel.
\end{description}
\end{itemize}

\item[II.]
At most one index is in $I_1$, so there must be indices in $I_2$ and $I_3$ that are adjacent in the circular order.
Without loss of generality let them be 
$1\in I_2, 2\in I_3$, 
i.e., $C(1,2)$ is small, $S_1\subseteq X$, and $S_2\cap X=\emptyset$. By Lemma~\ref{lem1} there exist paths from any vertex in $S_2$ to $C_2\backslash X$, $C_3\backslash X$, and $T\backslash X$ that are not blocked by $X$, so they are all on the same side of $X$, call it side $K$.
Refer to Figure~\ref{fig:Claim7-8} in Claims 7 and 8.

\begin{description}
    \item[Claim 7.] There exists another side $L$ of $X$, such that $L\cap C_1\neq \emptyset$.
    
        Note $C(1,3)$ separates $D(1,3)$ from the rest of the graph, and $X$ separates $K$ from the rest of the graph.
    If the claim were not true, then $C_1\subseteq X\cup K$, so 
    $G\backslash(X\cup K\cup C(1, 3)\cup D(1, 3))\subseteq G\backslash (X\cup K)\neq \emptyset$.
    We can now apply Lemma~\ref{lem2} (union rule) to $C(1,3)$ and $X$, and deduce that 
    $K\cup D(1, 3)$ is disconnected from the rest of the graph by
    \[
    (C(1, 3)\backslash K)\cup (X\backslash D(1, 3))\subseteq X\backslash S_1.
    \]
    But $\abs{X\backslash S_1} < \abs{X} = \kappa$, a contradiction. So there exists a side $L$ such that $L\cap C_1\neq \emptyset$.

\begin{figure}[h!]
\centerline{\scalebox{.30}{\includegraphics{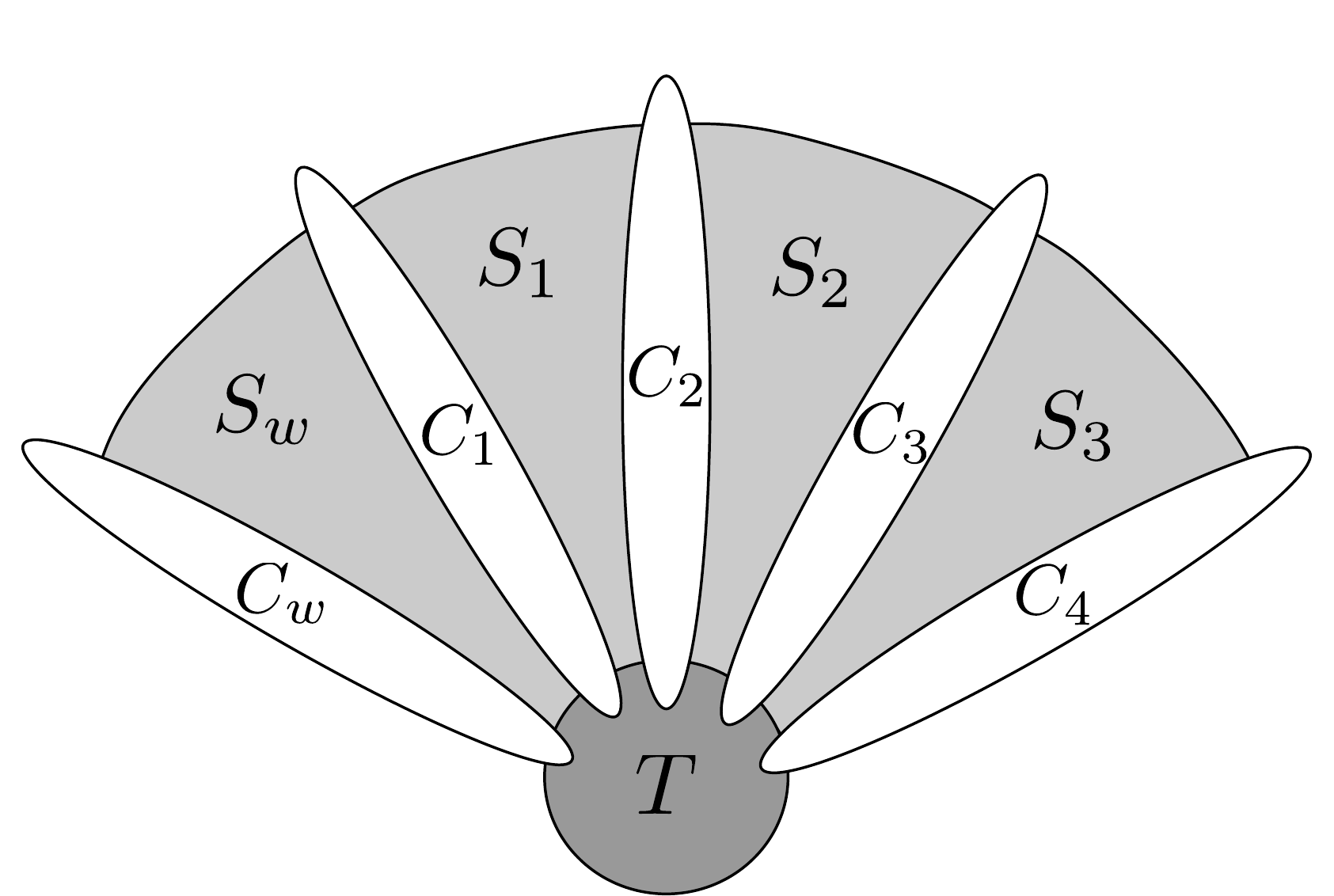}}}
\caption{In case II $S_1\subset X$ and $S_2\cap X=\emptyset$.}\label{fig:Claim7-8}
\end{figure}

    \item[Claim 8.] $X$ is a crossing matching cut of $C(1, 3)$ in side $D(1, 3)$ w.r.t.~$L\cap C_1$, or a crossing matching cut of $C(w, 2)$ in side $D(w, 2)$ w.r.t.~$K\cap C_2$.
    
        If $X\cap D(3, 1)\neq \emptyset$ 
        then Corollary~\ref{cor2} implies that
        $X$ is a crossing matching cut of $C(1, 3)$ in side $D(1, 3)$ w.r.t.~$L\cap C_1$.


        Otherwise, if $X\cap D(3, 1)=\emptyset$, then because there are paths from every vertex in $D(3, 1)$ to $S_w$ within itself, and paths from $S_w$ to $C_1\cap L$ not blocked by $X$, $D(3, 1)\subseteq L$. Now that $S_3\subseteq L$ and $S_2\subseteq K$, it follows that $C_3\cup T\subseteq X$.
        Based on the same reason of Claim 7, we have $K\cap C_2\neq \emptyset$. Now for cut $C(w, 2)$, we know that $X\cap D(w, 2)\supseteq S_1\neq \emptyset$, $X\cap D(2, w)\supset C_3\neq \emptyset$, $D(w, 2)=S_w\cup C_1\cup S_1\subseteq X\cup L$, and $K\cap C(w, 2)= K\cap C_2\neq \emptyset$, by 
        Corollary~\ref{cor2}, $X$ is a crossing matching cut of $C(w, 2)$ w.r.t.~$K\cap C_2$.
u
\end{description}

\item [III.] W.l.o.g.~we assume $I_1=\{1\}$ and therefore $I_3=\{2, 3, \ldots, w\}$.
With $Y_*,Z_*$ defined as in case I, we have
\[
X = Y_1\cup Q\cup (\cup_{i=1}^w Z_i). 
\]
Because $X$ has a 
wheel type relation
with $C(1,2)$,
$X$ has exactly two sides,
say $K$ and $L$,
and moreover,
the intersections of $X$ with the two sides
of $C(1,2)$ have equal size, i.e.,
\[
\abs{Y_1}=\sum_{r\not\in\{1, 2\}}\abs{Z_r}=\frac{\kappa -\abs{Q}-\abs{Z_1}-\abs{Z_2}}{2}.
\]
\begin{description}
    \item[Claim 9.] $Q=T$.

If $T\backslash X\neq \emptyset$, 
by Lemma~\ref{lem1} there exists paths 
from every vertex in $S_r$, $r\neq 1$,
to $T\backslash X$, 
to $C_r\backslash X$, 
and to $C_{r+1}\backslash X$, 
that are not blocked by $X$.
Thus, 
$((D(2, 1) \cup T)\backslash X)$ 
is within a side of $X$, say side $K$. 
It follows that the other side
$L$ satisfies $L \subseteq S_1$.
Applying Lemma~\ref{lem2} (intersection rule)
to $C(1,2)$ and $X$, we find that 
$L\cap S_1=L$ is disconnected 
from the rest of the graph by
\[
(L\cap C(1, 2))\cup (X\cap C(1, 2))\cup(S_1\cap X)=Y_1\cup Q\cup Z_1\cup Z_2.
\]
%
That means
\[
\abs{Y_1}+\abs{Q}+\abs{Z_1}+\abs{Z_2}\geq \kappa,
\]
which implies
\[
\kappa = \abs{X} = \abs{Y_1}+\abs{Q}+\sum_{i=1}^w \abs{Z_i} \geq \abs{Y_1} + \abs{Q} + \abs{Z_1} + \abs{Z_2} \geq \kappa,
\]

and therefore $Z_r=\emptyset$ for all $r\not\in\{1, 2\}$, 
so $X\subseteq(C_1\cup T\cup C_2\cup S_1)$, contradicting the fact that 
$X,C(1,2)$ have wheel type.  Thus $Q=T$.

\item[Claim 10] There exists $j\not\in \{1, 2\}$ such that $Z_j=C_j$. 

If there does not exist such a $j$, 
then it follows that $C_r\backslash X$ is nonempty for every $r\not\in\{1, 2\}$. 
Thus, there are paths from any vertex in 
$S_r$ to $C_r\backslash X$ and to $C_{r+1}\backslash X$
that are not blocked by $X$, 
so $D(2, 1)\backslash X$ is still contained in a side of $X$,
say $K$. 
If $C_1\backslash X$ or $C_2\backslash X$ is nonempty, then they are also in $K$ as well, since there exist paths from 
$C_1\backslash X$ to $S_w$, and from $C_2\backslash X$ to $S_2$. 
Then we have $L\subseteq S_1$. 
Based on exactly the same reasoning in Claim 9 we obtain 
the contradiction that 
$X\subseteq (C_1\cup T\cup C_2 \cup S_1)$. 
This proves the claim.

\item[Claim 11] $(T; C_1, Y_1, C_2, \ldots, C_w)$ forms a wheel.

Since $Q=T$, $C_j=Z_j$, and
\[\sum_{r\not\in\{1, 2\}}\abs{Z_r} = \frac{\kappa - \abs{Q}-\abs{Z_1}-\abs{Z_2}}{2}\leq \frac{\kappa-\abs{Q}}{2}=\frac{\kappa-\abs{T}}2 = \abs{C_j}=\abs{Z_j},\]
it follows that $\abs{Z_i}=0$ for $i\neq j$.
By Lemma~\ref{lem2} (intersection rule) 
it is easy to verify that 
$(T; C_1, Y_1, C_2, \ldots, C_w)$ forms a $(w+1)$-wheel,
as is done in Claims 3 and 6, and this puts us in
case 6\cir.

\end{description}

\item[IV.]
Recall that $I_1\cup I_2\cup I_3=\{1, 2, \ldots, w\}$. In this case, except for at most one $i^*\in I_1$, all other $i$ are in $I_2$ and therefore $S_i\subseteq X$, then $\sum_{i\neq i^*}S_i\leq \abs{X}=\kappa$, so this is a small wheel
and we are in case 5\cir.

\item[V.]
Again recall that $I_1\cup I_2\cup I_3=\{1, 2, \ldots, w\}$. So in this case $I_3=\{1, 2, \ldots, w\}$. Therefore, $X\subseteq T\cup (\cup_{i=1}^w C_i)$. Note that the graph induced by 
$(\cup_{i=1}^w S_i)\cup \{t\}$ is connected for any $t\in T$, and that every 
vertex in some $C_i$ is adjacent to $S_{i-1}$ and $S_i$.
This implies that $T\subseteq X$, and that there exists indices 
$i\neq j$ such that $C_i\cup C_j\subseteq X$; 
otherwise $X$ is not even a cut. 
We have deduced that $X=C(i,j)$ for some $i,j$, putting us in case $1$\cir.
\end{itemize}
\end{proof}

\ignore{
\begin{remark}
This theorem completely solves all problems about wheels. Given a wheel that is not small, either it is a subwheel of another wheel, or all other cuts are independent from it or crossing matching cut of some $C_i\cup T\cup C_{i+1}$ or $C_{i}\cup T\cup C_{i+2}$. We may split the graph into $S_i$ and search within $S_i\cup C_i\cup T\cup C_{i+1}$ all the cuts that are independent from the wheel, and all the crossing matching cuts.

When a wheel is a small wheel, all the $C_i\cup T\cup C_{i+2}$ are (\rmnum{2}, $2\kappa$) small. At this time the pairs that can be disconnected by the cuts of the wheel are disconnected by the small cuts found, so we do not need to maintain small wheels.

Given a graph $G$, we now can define maximal wheels. That is, there does not exists other cuts which makes $6$\cir happen. In other words, if $6$\cir happens, we may find a larger wheel. When we have a maximal wheel, then all the other cuts may be a laminar cut within a sector, or a small cut, or just a crossing matching cut.
\end{remark}
}

\subsection{Matching Cuts and Crossing Matching Cuts}\label{subsect:matching}

Define $N(P)$ to be the neighborhood of $P\subset V$
and $N_A(P) \bydef N(P)\cap A$.

\begin{theorem}
  \label{lem4}
  Let $U$ be an arbitrary $\kappa$-cut and $A$ a side of $U$.
  \begin{itemize}
    \item [$1$\cir] 
    If there exists a matching $\kappa$-cut of $U$ in side $A$ w.r.t.~$P$, then it is $W = (U\backslash P)\cup N_A(P)$,
    and $\abs{P}=\abs{N_A(P)} < \abs{A}$.  In particular,
    $\Match{U;A}(P) = N_A(P)$.
    
    \item [$2$\cir ] When there is such a matching cut $W$, 
    $G$ contains a matching between $P$ and $\Match{U;A}(P)=N_A(P)$.
    
    \item [$3$\cir] Suppose 
    $\Match{U;A}(P)$ and $\Match{U;A}(Q)$ exist. 
    If $P\cap Q\neq \emptyset$, then 
    $\Match{U;A}(P\cap Q)$ exists, and
        \[\ \Match{U;A}(P\cap Q)=\Match{U;A}(P)\cap \Match{U;A}(Q).\]
    If $\abs{A}>\abs{P\cup Q}$, then $\Match{U;A}(P\cup Q)$ 
    exists, and
        \[\Match{U;A}(P\cup Q)=\Match{U;A}(P)\cup \Match{U;A}(Q).\]
  \end{itemize}
\end{theorem}

\begin{proof}
Part 1\cir.
By definition, 
$W$ is a matching cut in side $A$ w.r.t.~$P$
if 
(i) $W\cap U = U\backslash P$, 
(ii) $W\subseteq U\cup A$,
and 
(iii) $W$ separates every vertex in $A\backslash W$ from $P$.
It must be that $N_A(P) \subseteq W$, for otherwise $W$ would 
not satisfy (iii).  It also follows from (iii) 
that $N_A(P)\backslash A \neq \emptyset$.
Since $W' = (U\backslash P)\cup N_A(P)$ is a cut
(separating $A\backslash N_A(P)$ 
from $P\cup V\backslash (U\cup A)$),
it follows that
$W'\subseteq W$,
but since 
$W$ is a (minimum) $\kappa$-cut, 
then $W=W'$ 
and hence $\abs{P}=\abs{N_A(P)}$.
Thus, by definition 
$\Match{U;A}(P) = N_A(P) = W\backslash U$
is just the neighborhood function $N_A(P)$
whenever such a matching cut $W$ exists.

\medskip

Part 2\cir.
Define $H$ to be the bipartite subgraph of $G$
between $P$ and $N_A(P)=\Match{U;A}(P)$.
By Hall's theorem, if $H$ does not contain a matching
then there exists a strict subset $P'\subset P$
such that $\abs{P'}>\abs{N_A(P')}$, 
but if this were the case,
$(U\backslash P')\cup N_A(P')$ would be a cut with cardinality
strictly smaller than $\kappa$, a contradiction.

\medskip 

\begin{figure}[h!]
    \centering
    \begin{tabular}{c@{\hspace*{2cm}}c}
    \scalebox{.5}{\includegraphics{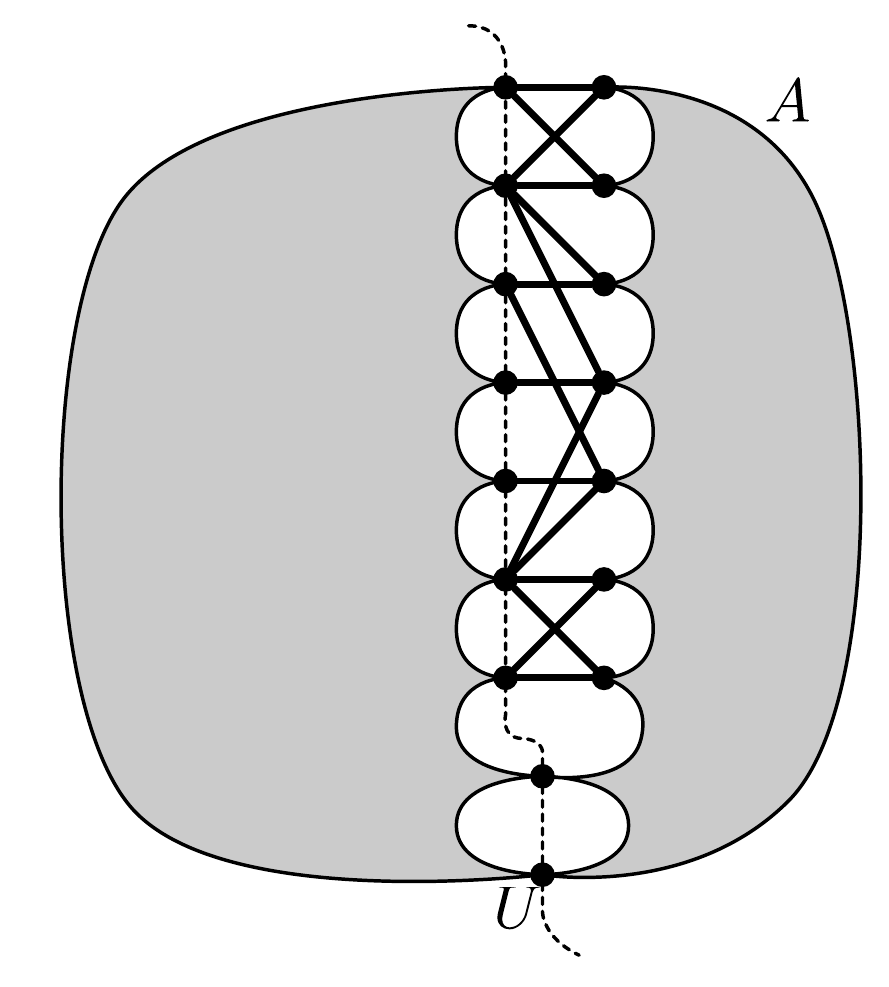}}
    &
    \scalebox{.5}{\includegraphics{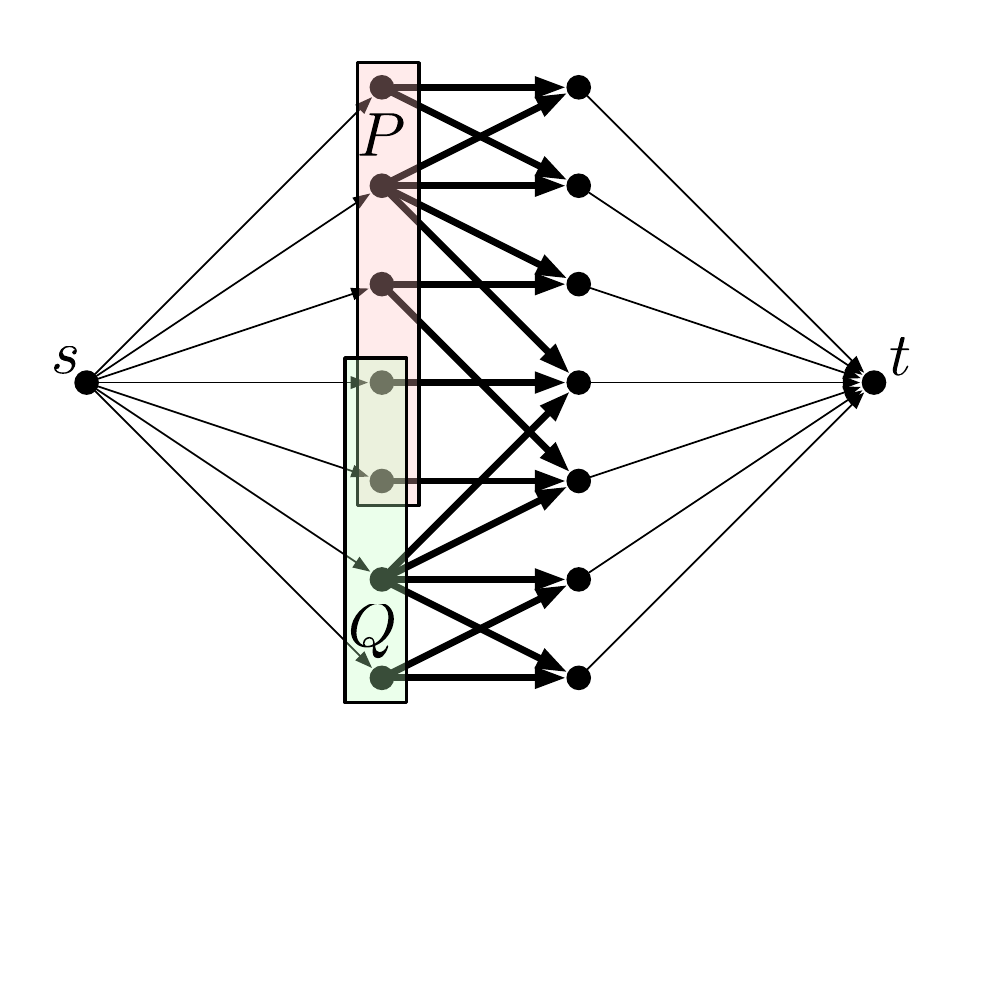}}
    \end{tabular}
    \caption{Left: a cut $U$ with matching cuts in side $A$.  Right: the derived flow network for $P,Q\subset U$.}
    \label{fig:Matching-PQ}
\end{figure}
Part 3\cir.
Redefine $H$ to be the bipartite subgraph of $G$
between $P\cup Q$ and $N_A(P\cup Q)$.
We construct a directed flow network
$\vec{H}$ on $\{s,t\}\cup V(H)$ as follows; see Figure~\ref{fig:Matching-PQ}.
All edges from $H$ appear in $\vec{H}$, 
oriented from $P\cup Q$ to $N_A(P\cup Q)$, 
each having infinite capacity.
The edge set also includes edges from $s$ to $P\cup Q$
and $N_A(P\cup Q)$ to $t$, each with unit capacity.
Clearly integer flows in $\vec{H}$ 
correspond to matchings in $H$.
By the K\H{o}nig-Egerv\'{a}ry theorem, 
the size of the minimum vertex cover in $H$ 
is equal to the size of the maximum matching in $H$.
Following the same argument in 2\cir, 
they must both be of size $\abs{P\cup Q}$ for otherwise
$U$ would not be a minimum vertex cut.
In fact, minimum size
vertex covers of $H$ are in 1-1 correspondence 
with minimum capacity $s$-$t$ cuts in $\vec{H}$.
The correspondence is as follows.
If $(S,T)$ is
a minimum capacity
$s$-$t$ cut
then there can be no (infinite capacity) edge 
from $S\cap (P\cup Q)$
to $T\cap N_A(P\cup Q)$, so
\begin{align*}
C &=((P\cup Q)\cap T) \cup (N_A(P\cup Q)\cap S)
\intertext{is a vertex cover in $H$ and}
\abs{C} &= \operatorname{cap}(S,T).
\end{align*}

Picard and Queyrenne~\cite{PicardQ80} observed
that minimum $s$-$t$ cuts are closed under union and intersection,
in the sense that if $(S,T),(S',T')$ are min $s$-$t$ cuts,
then so are $(S\cup S',T\cap T')$ and $(S\cap S',T\cup T')$.
By assumption $(P\backslash Q)\cup N_A(Q)$ and $(Q\backslash P)\cup N_A(P)$ 
are minimum vertex covers\footnote{This follows from the fact that
$(U\backslash Q)\cup N_A(Q)$ 
and $(U\backslash P)\cup N_A(P)$ are
assumed to be matching cuts of $U$ in $A$ w.r.t.~$Q$ and $P$,
respectively,
with $\Match{U;A}(P)=N_A(P)$
and $\Match{U;A}(Q)=N_A(Q)$.} 
of $H$, with cardinality $\abs{P\cup Q}$.
Translated to the flow network $\vec{H}$, 
this implies that 

\begin{tabular}{@{\hspace*{2cm}}c@{\hspace*{2cm}}c}
    \scalebox{.5}{\includegraphics{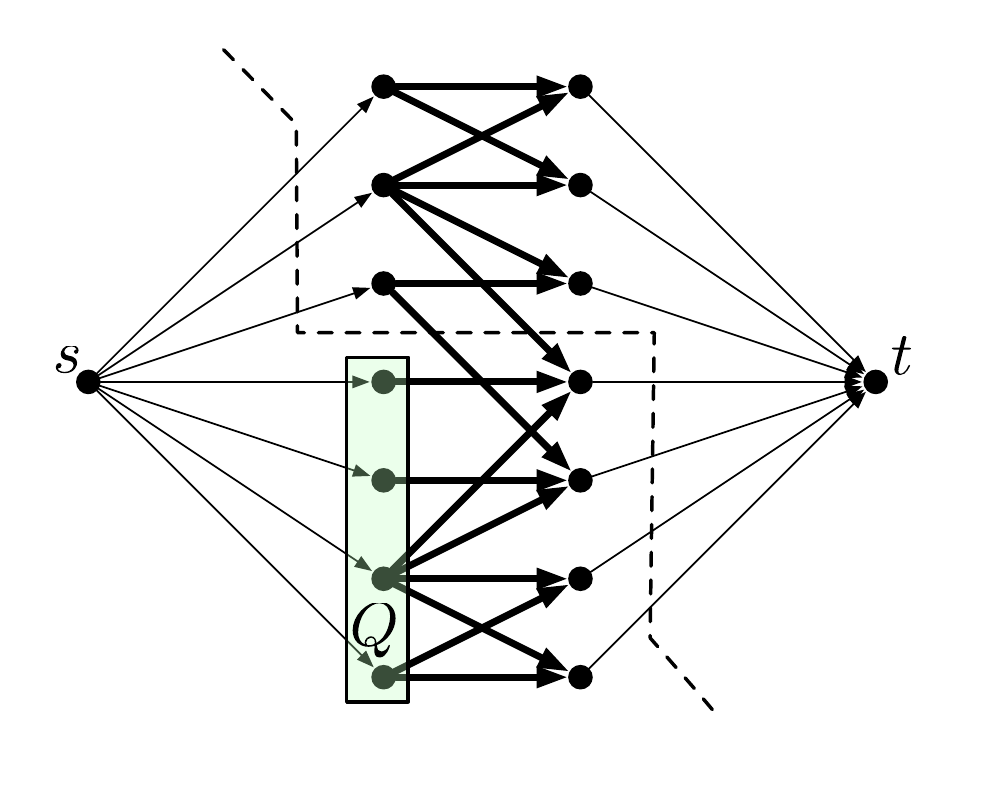}}
    &
    \scalebox{.5}{\includegraphics{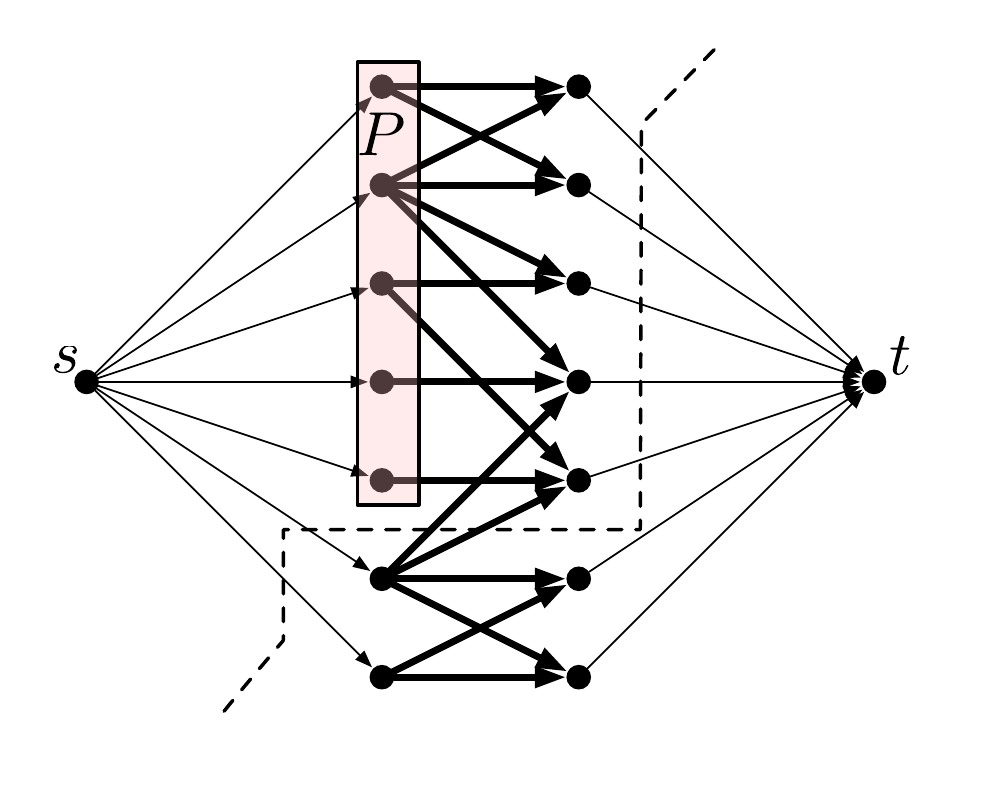}}
\end{tabular}
    
\begin{align*}
&(\{s\}\cup Q \cup N_A(Q), \; (P\backslash Q) \cup N_A(P\cup Q)\backslash N_A(Q)\cup \{t\})\\
\mbox{and } &(\{s\}\cup P \cup N_A(P), \; (Q\backslash P) 
\cup N_A(P\cup Q)\backslash N_A(P) \cup \{t\})
\intertext{are minimum $s$-$t$ cuts, which implies that their union and intersection are also minimum $s$-$t$ cuts:
}
&(\{s\}\cup P\cup Q \cup N_A(P\cup Q), \; \{t\})\\
\mbox{and } &\bigg(\{s\} \cup (P\cap Q)\cup (N_A(P)\cap N_A(Q)), \; \\
&\Big((P\cup Q)\backslash (P\cap Q)\Big) \cup 
\Big(N_A(P\cup Q)\backslash (N_A(P)\cap N_A(Q))\Big)\cup \{t\}\bigg).
\end{align*}

\begin{tabular}{@{\hspace*{2cm}}c@{\hspace*{2cm}}c}
    \scalebox{.5}{\includegraphics{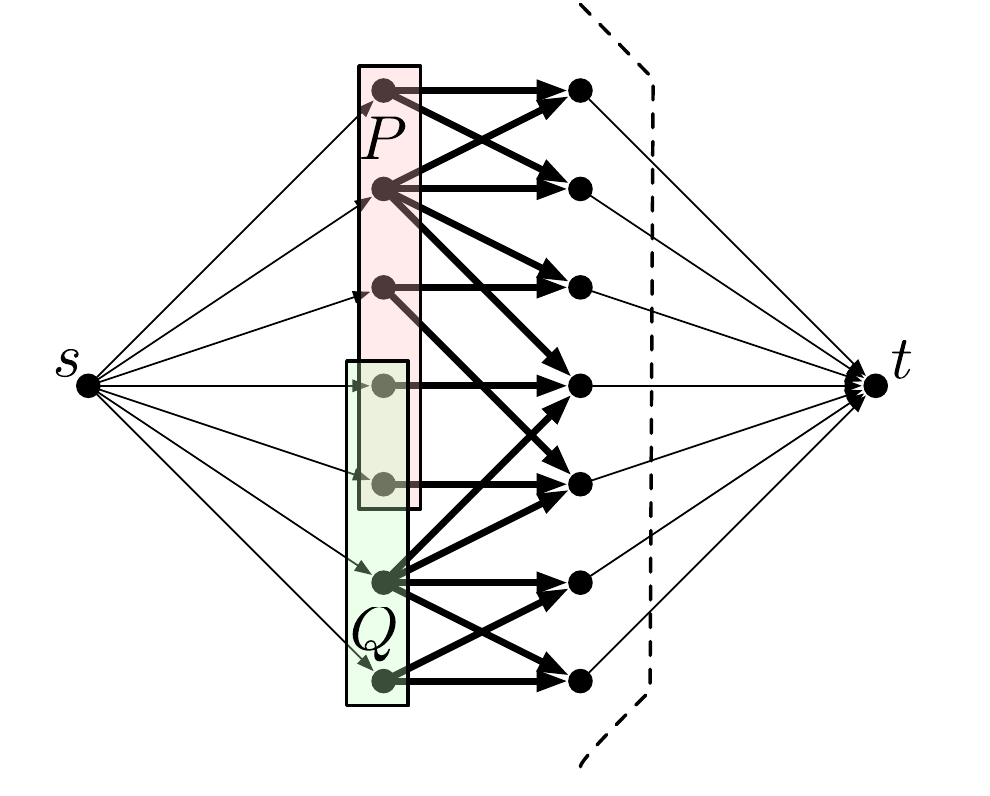}}
    &
    \scalebox{.5}{\includegraphics{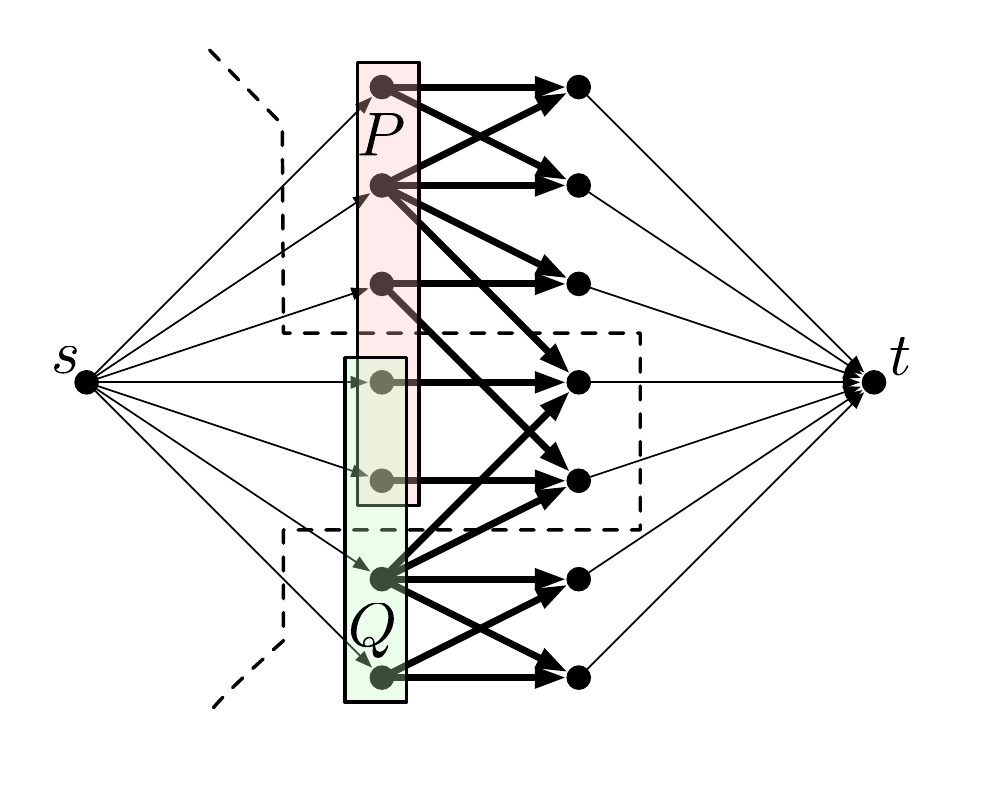}}
\end{tabular}
    
\noindent Translated back to $H$, these 
correspond to vertex covers
\begin{align*}
& N_A(P\cup Q)\\
\mbox{and }
& (P\cup Q)\backslash (P\cap Q) \cup (N_A(P)\cap N_A(Q)).
\end{align*}
Whenever $N_A(P\cup Q)$ is a \emph{strict} subset of $A$
the first corresponds to a matching $\kappa$-cut of $U$
\begin{align*}
& (U\backslash (P\cup Q))\cup N_A(P\cup Q)\\
\mbox{ with \ } 
\Match{U;A}(P\cup Q) &= N_A(P\cup Q) = \Match{U;A}(P)\cup \Match{U;A}(Q).
\intertext{Whenever $P\cap Q\neq \emptyset$ the second 
corresponds to a matching $\kappa$-cut of $U$}
(U\backslash (P\cap Q))\cup (N_A(P)\cap N_A(Q))\\
\mbox{ with \ }
\Match{U;A}(P\cap Q) 
&= N_A(P)\cap N_A(Q) = \Match{U;A}(P)\cap \Match{U;A}(Q).
\end{align*}
\end{proof}

Fix a $\kappa$-cut $U$ and a side $A$ of $U$.
Define $\Theta = \{P \mid \Match{U;A}(P) \mbox{ exists}\}$.
According to Part 3\cir\ of Theorem~\ref{lem4}, $\Theta$
is closed under union and intersection, and
is therefore characterized by its minimal elements.
Define $\Theta^* = \{\cap_{u\in P, P\in \Theta}P \mid u\in U\}$.
It can be seen from the definition that $\cap_{u\in P, P\in \Theta}P$
corresponds to the minimum matching cut for vertex $u$.

In the most extreme case 
$\Theta$ may have $2^\kappa - 1$
elements (e.g., if the graph induced
by $U\cup N_A(U)$ is a matching), 
which may be prohibitive to store explicitly. 
From definition we know that $\abs{\Theta^*}\leq \kappa$,
so it works as a good compression for $\Theta$.
Lemmas~\ref{lem6} and \ref{lem7}
also highlights some ways in which $\Theta^*$
is a sufficient substitute for $\Theta$.

\begin{lemma}
\label{lem6}
Let $U$ be a $\kappa$-cut
and let $\Theta$ be defined w.r.t.~the matching cuts of $U$
in a side $A$.
Suppose that $P\in \Theta^*$ and 
$P\subseteq Q\in\Theta$,
and that $W$ is a crossing matching 
cut of $U$ in side $A$
w.r.t.~$Q$.  Then 
$(W\backslash \Match{U;A}(Q))\cup (Q\backslash P)\cup \Match{U;A}(P)$ is also a crossing matching cut of $U$ in $A$ w.r.t.~$P$.
Moreover, if $Q=P_1\cup P_2\cup \cdots \cup P_\ell$ where each 
$P_i\in \Theta^*$, then any pair disconnected by $W$ is also disconnected by 
some $(W\backslash \Match{U;A}(Q))\cup (Q\backslash P_i)\cup \Match{U;A}(P_i)$.
\end{lemma}

\begin{proof}
Apply Lemma~\ref{lem2} (intersection rule) to cuts
$W$ and $(U\backslash P)\cup \Match{U;A}(P)$
and we have the first statement, that
$(W\backslash \Match{U;A}(Q))\cup (Q\backslash P)\cup \Match{U;A}(P)$
is a cut.

For the second statement, suppose $W$ separates vertices 
$u$ and $v$.  If both $u,v \not\in U\cup \Match{U;A}(Q)$,
then they are separated by \emph{any} 
$(W\backslash \Match{U;A}(Q))\cup (Q\backslash P_i)\cup \Match{U;A}(P_i)$.
If one of $u,v$ is in $Q$, say $u\in Q$, then there exists at least one
$i$ for which $u\in P_i$.  Then $v\in A\backslash W$ and therefore
$u,v$ are separated by $(W\backslash \Match{U;A}(Q))\cup (Q\backslash P_i)\cup \Match{U;A}(P_i)$.
\end{proof}

\begin{lemma}
  \label{lem7}
  Let $U$ be a $\kappa$-cut with two sides $A$ and $B$,
  and let $\Theta$ be defined w.r.t.~its matching cuts
  in side $A$.
  For $P\in \Theta^*$, define $U^*(P)$ to be the 
  cut separating $P$ from $A\backslash \Match{U;A}(P)$
  minimizing $\abs{\Side{U^*(P)}(P)}$.
\begin{enumerate}
    \item[1\cir] $U^*(P)$ is either a crossing matching cut of $U$ in side $A$ w.r.t.~$P$, or else there is no such crossing matching cut and 
    $U^*(P) = (U\backslash P)\cup \Match{U;A}(P)$ is a matching cut.
    \item[2\cir] Suppose $X$ is a 
    crossing matching cut of $U$ in side $A$
    w.r.t.~$P$.  If $u,v$ are separated by 
    $X$, then they are also separated by either
    $U^*(P)$
    or $(X\backslash \Match{U;A}(P))\cup P$, which is a laminar cut of $U$.
\end{enumerate}
\end{lemma}
\begin{proof}
Part 1\cir.  
Let $W$ be a cut that disconnects $P$ from $A\backslash \Match{U;A}(P)$. 
Because every vertex in $\Match{U;A}(P)$ is 
adjacent to $P$ and adjacent to 
$A\backslash \Match{U;A}(P)$, it follows that
$\Match{U;A}(P)\subseteq W$.
From Corollary~\ref{cor:CD} we know the set of all cuts separating $P$ from $A\backslash \Match{U;A}(P)$ has a unique minimum element.
This is the cut $U^*(P)$; let the sides of $U^*(P)$ be $K^*,L^*$ with $P\subseteq K^* \subseteq B\cup P$ and $A\backslash \Match{U;A}(P)\subseteq L^*$.

If $K^*=B\cup P$, then $U^*(P) = (U\backslash P) \cup \Match{U;A}(P)$. 
If $K^*\subset B\cup P$, 
then 
$U^*(P)\cap B\neq \emptyset$. 
Now since $\Match{U;A}(P)\subseteq U^*(P)$
and 
$(U^*(P)\cap A)\cup (U\backslash P)$ 
is a matching cut w.r.t.~$P$,
it follows that 
$U^*(P)$ is a crossing matching cut of 
$U$ in side $A$ w.r.t.~$P$.

Part 2\cir.
Fix such a crossing matching cut $X$. 
It disconnects $P$ from $A\backslash \Match{U;A}(P)$,
and has exactly two sides, 
$K \supseteq P$ and 
$L\supseteq A\backslash \Match{U;A}(P)$.
By the minimality of $U^*(P)$ we 
have $K^*\subseteq K$.
Because $X\cap B\neq \emptyset$, $K\subset B\cup P$. Thus, if $U^*(P)=(U\backslash P)\cup \Match{U;A}(P)$ is the matching cut (not a crossing matching cut), this contradicts the existence of $X$.
We conclude that $U^*(P)$ is a \emph{crossing} matching cut of $U$.

Let $u\in K$ and $v\in L$ be disconnected by $X$. 
If $u\in B\cap K$ and $v\in L$
then by Lemma~\ref{lem2} (intersection rule) 
applied to $U$ and $X$,
$B\cap K$ is disconnected from the rest of the graph by $(X\cap B)\cup (X\cap U)\cup(U\cap K)=(X\backslash \Match{U;A}(P))\cup P$, 
which is a laminar cut of $U$.
By the minimality of $U^*(P)$, 
$K^*\subseteq K$ and $U^*(P)\subseteq K\cup X$, 
hence $L\subseteq L^*$. 
When $u\in P$ and $v\in L$, 
they are also disconnected by $U^*(P)$.
\end{proof}

\begin{corollary}
  \label{thm3}
  Fix a cut $U$ with two sides $A$ and $B$, and let $\Theta^*$ 
  be defined w.r.t.~the matching cuts in side $A$, 
  and let $U^*(P)$ be defined as in Lemma~\ref{lem7}.
  Define $\mathcal{U} = \{U^*(P) \mid P\in \Theta^*\}$.

  Let $X$ be a crossing matching cut of $U$ in side $A$ w.r.t.~$Q$.  If $u$ and $v$
  are separated by $X$, then they are also separated by a member of $\mathcal{U}$
  or $(X\backslash \Match{U;A}(Q))\cup Q$, a laminar cut of $U$ in side $B$;
\end{corollary}

\begin{proof}
By Lemma $\ref{lem6}$, w.l.o.g.~we may assume that $Q\in \Theta^*$. 
Since $X$ is a crossing matching cut,
there is a crossing matching cut 
$U^*(Q)\in \mathcal{U}$. 
The claim then follows from Lemma~\ref{lem7}.
\end{proof}

\subsection{Laminar Cuts}\label{subsect:laminar}

In this section we analyze the structure of laminar cuts.
Throughout this section, $U$ refers to a cut that 
is \emph{not} (\Rmnum{1}, $\kappa-1$)-small,
\emph{not} a wheel cut $C(i,j)$ in some wheel,
and has a side $A$ with $\abs{A}>2\kappa$.

Consider the set of all cuts 
$W$ that are laminar w.r.t.~$U$, contained in $U\cup A$ and not (\Rmnum{1}, $\kappa-1$)-small.
It follows that $W$ has a side, call it $S(W)$,
that contains $U\backslash W$ and all other sides of $U$.\footnote{$S(W)$ is exactly $B_{j^*}$ of Theorem~\ref{thm1}, if using its notation on $U$ and $W$.}
Define $R(W)$ to be the region containing all other sides of $W$ beside $S(W)$.
We call $W$ a \emph{maximal} laminar cut of $U$ if there does not exist
another laminar cut $W'$ such that $R(W)\subseteq R(W')$.

\begin{theorem}\label{thm5}
Let $U$ be the reference cut.
  \begin{enumerate}
    \item If there exist matching cuts of $U$ in side $A$, 
    define $\Theta^*$ w.r.t.~$U,A$,
    define $Q=\cup_{P\in \Theta^*}P$, 
    and let $X=(U\backslash Q)\cup \Match{U;A}(Q)$ be the 
    matching cut in side $A$ having the smallest intersection with $U$.
    Then every laminar cut $W$ of $U$ in side $A$ is
  \begin{enumerate}
    \item [(\rmnum{1})] a laminar cut of $X$ in region $A\backslash \Match{U;A}(Q)$, or
    \item [(\rmnum{2})] a matching cut of $U$, or
    \item [(\rmnum{3})] a crossing matching cut of $X$.
  \end{enumerate}
  \item If there are no matching cuts of $U$ in side $A$, 
  every laminar cut of $U$ in side $A$
  is a maximal laminar cut, or a laminar
  cut of some maximal laminar cut $W_i$ in 
  a side of $R(W_i)$.
    Moreover, whenever $W_i$, $W_j$ are distinct maximal laminar cuts,
    $R(W_i)\cap R(W_j)=\emptyset$.
  \end{enumerate}
\end{theorem}

\begin{proof}
Part 1. $X$ has a side $K=Q\cup (V\backslash(U\cup A))$ and a region $L=A\backslash \Match{U; A}(Q)$. According to the number of sides of $X$, we split into two cases:
\begin{itemize}
\item[\Rmnum{1}.] If $X$ has strictly more than two sides, let its sides be $K$ and $L_1$, $L_2$, $\ldots$, $L_r$ where $L=\cup_{i=1}^r L_i$. Then by Corollary~\ref{cor1}, any other cut should have laminar type relation with $X$, 
or themselves be (\Rmnum{1}, $\kappa-1$)-small cuts.
But here the cut $W$ in our concern are not (\Rmnum{1}, $\kappa-1$)-small.
For such a $W$ of $U$ in side $A$, it can only be a laminar cut of 
$X$ in some side. If it was $L_i$, then we have (\rmnum{1}). 
If it was $K$, then we have that $W\subseteq X\cup K$ and also $W\subseteq U\cup A$, so that $W\subseteq (X\cup K)\cap (U\cup A)=U\cup \Match{U; A}(Q)$. Then by definition of $\Theta$, $W$ is a matching cut of $U$ and we have (\rmnum{2}).

\item[\Rmnum{2}.]
If $X$ has exactly two sides, 
then they are $K=Q\cup (V\backslash (U\cup A))$ and $L=A\backslash \Match{U;A}(Q)$. 
Fix any laminar cut 
$W$ of $U$ in side $A$.
If $W\subseteq A\cup (U\backslash Q)$ 
then $W$ is a laminar cut of $X$ in side $A\backslash \Match{U;A}(Q)$, 
and we are in case (i).
If $W\subseteq U\cup \Match{U;A}(Q)$, 
then by definition of $\Theta$,
$W=(U\backslash P)\cup \Match{U;A}(P)$ 
for some $P\in \Theta$ and $W$ 
is a matching cut of $U$, and we are 
in case (ii).

Thus, we can proceed with the assumption
that $W\cap Q\neq \emptyset$ and $W\cap (A\backslash \Match{U; A}(Q))\neq \emptyset$. Therefore, $W\cap K\neq \emptyset$ and $W\cap L\neq\emptyset$. So $W$ must have wheel type or crossing matching type relation with $X$. We only need to show that $W$ does not have wheel type relation with $X$.
Because $R(W)\subseteq A$, we know that $R(W)\cap K\subseteq A\cap K=\emptyset$, so they do not have a wheel type relation. 
Then, $W$ must be a crossing matching cut of $X$, and we are in case (iii).
\end{itemize}

\medskip

Part 2.
By assumption $U$ does not have matching cuts in side $A$.
Enumerate all of its maximal laminar 
cuts $\mathcal{W} = \{W_1, W_2, \ldots, W_{|\mathcal{W}|}\}$.
Fix any laminar cut $W$ of $U$ in side $A$.
If $W\not\in\mathcal{W}$,
then by definition of maximality 
there exists 
some $W_i$ such that 
$R(W)\subseteq R(W_i)$.
It follows that 
$W \subseteq R(W_i)\cup W_i$, 
so $W$ is a laminar cut of $W_i$ 
in one of the sides of $R(W_i)$. 

It remains to prove that 
for all $i\neq j$, 
$R(W_i)\cap R(W_j)=\emptyset$.
Suppose the statement were false. 
Because 
$S(W_i) \cap S(W_j)\neq \emptyset$, 
$W_i,W_j$ must have laminar, 
crossing matching, or wheel type relation.
They cannot be laminar, for then
$R(W_i)\subset R(W_j)$, 
or 
$R(W_j)\subset R(W_i)$, 
contradicting the maximality of $W_i,W_j$.
Otherwise, $W_i$ and $W_j$ must both 
have exactly two sides, namely 
$R(W_i),S(W_i)$ and $R(W_j),S(W_j)$.
Apply Corollary~\ref{cor:CD} to $C=R(W_i)\cap R(W_j)$ and $D=S(W_i)\cap S(W_j)$, noticing that $W_i$ and $W_j$ disconnects $C$ and $D$, we may set $Y=\MinCut{C; D}$.
Because $\Region{Y}(C)\subseteq \Region{W_i}(C)\cap \Region{W_j}(C)=R(W_i)\cap R(W_j)=C$, $C$ is actually a region of $U$.
As long as $R(W_i)\cup R(W_j)\neq A$, 
$Y$ is a laminar cut of $U$,
also contradicting the maximality of $W_i,W_j$.
Thus, we proceed under the assumption that
$R(W_i) \cup R(W_j)=A$, meaning
$Y=U$ is not a laminar cut of $U$.

If $W_i$ and $W_j$ have a crossing matching type relation, at least one of $R(W_i)\cap S(W_j)$ and $R(W_j)\cap S(W_i)$ is empty, suppose it is 
$R(W_i)\cap S(W_j)=\emptyset$. 
Then $R(W_i)\subseteq R(W_j)\cup W_j$, but we already have that $R(W_i)\cup R(W_j)=A$, so $R(W_j)=A\backslash W_j$, 
which means $W_j$ is a matching cut of $U$ in side $A$, 
contradicting the assumption of Part 2 
that $U$ has no matching cuts in side $A$.

The last case is when
$W_i$ and $W_j$ have a wheel type relation, 
i.e., they form a $4$-wheel 
with center $T = W_i\cap W_j$, spokes 
$W_i\cap R(W_j)$, 
$W_i\cap S(W_j)$, 
$W_j\cap R(W_i)$,
$W_j\cap S(W_i)$, and
sectors 
$R(W_i)\cap R(W_j)$, 
$R(W_i)\cap S(W_j)$, 
$S(W_i)\cap R(W_j)$, 
$S(W_i)\cap S(W_j)$. 
Thus, the $\kappa$-cut 
$(W_i\cap S(W_j))\cup (W_i\cap W_j)\cup (W_j\cap S(W_i))$ disconnects $S(W_i)\cap S(W_j)$ from $R(W_i)\cup R(W_j)=A$. This means 
$U=(W_i\cap S(W_j))\cup(W_i\cap W_j)\cup (W_j\cap S(W_i))$ is also a cut of this wheel. 
This contradicts the original assumption that our reference cut
$U$ is not a wheel cut $C(i', j')$ of some wheel.
\end{proof}

\subsection{Small Cuts}\label{subsect:small}

Fix a vertex $u$ and a threshold $t\leq \ceil{\frac{n-\kappa}{2}}$.
Define $\Small{t}(u)$ to be a cut $U$ minimizing $\abs{\Side{U}(u)}$ 
with $\abs{\Side{U}(u)}\leq t$.
We first show that $\Small{t}(u)$, 
if it exists, is unique.

\begin{theorem}\label{thm:small-side-cut}
  If there exists any (\Rmnum{3}, $t$)-small cut that is small w.r.t.~$u$,
  then there exists a \emph{unique} 
  such cut, denoted $\Small{t}(u)$,
  such that 
  for any other cut $U$, $u\not\in U$, 
  \[\Side{\Small{t}(u)}\subseteq \Side{U}(u).\]
\end{theorem}
\begin{proof}
  It suffices to show that for any two
  (\Rmnum{3},$t$)-small cuts $U$, $W$,
    there exists a cut $X$ (possibly $U$ or $W$) such that 
\[
\Side{X}(u)\subseteq \Side{W}(u) \cap \Side{U}(u).
\]
  If $\Side{U}(u)\subseteq \Side{W}(u)$ or $\Side{W}(u)\subseteq \Side{U}(u)$, we may set $X=U$ or $X=W$. If $V\backslash(U\cup W\cup \Side{U}(u)\cup \Side{W}(u))\neq \emptyset$, we may pick an arbitrary vertex $v$ in this set, 
  and apply Corollary~\ref{cor:CD} to the singleton sets $C=\{u\}$ 
  and $D=\{v\}$, and we may set $X=\MinCut{C,D}$.

  These two cases above rule out the possibility that $U, W$
  have a laminar or wheel type relation, except when,
  using the notation of Theorem~\ref{thm1},
  $A_{i^*}=\Side{U}(u)$ and $B_{j^*}=\Side{W}(u)$. 
  But this would lead to a contradiction that
  \begin{align*}
      \abs{V} &= \abs{A_{i^*}} + \abs{B_{j^*}} - \abs{A_{i^*}\cap B_{j^*}} + \abs{U\cap W}\\
      &\leq 2\ceil{\frac{n-\kappa}{2}} - 1 + (\kappa - 1)\\
      &<n.
  \end{align*}
  
  By Theorem~\ref{thm1} the remaining case is
  that $U,W$ have crossing matching type, i.e.,
  $\Side{U}(u) \backslash \Side{W}(u) \neq \emptyset$, $\Side{W}(u) \backslash \Side{U}(u) \neq \emptyset$, 
  and $V=U\cup W\cup \Side{U}(u) \cup \Side{W}(u)$.
  By Lemma~\ref{lem2} (intersection rule), 
  $(U\cap \Side{W}(u))\cup (W\cap \Side{U}(u))\cup (U\cap W)$
  is a cut.  If it has size exactly $\kappa$ then we can set $X$
  to be this cut.  We proceed under the assumption that it is 
  strictly larger than $\kappa$.  Thus, by inclusion/exclusion,
\begin{align*}
    |V| &= |U\cup W\cup \Side{W}(u) \cup \Side{U}(u)|\\
    &= |U| + |W| + |\Side{W}(u)| + |\Side{U}(u)| \\
    &\hspace*{2cm} - |\Side{W}(u)\cap\Side{U}(u)| - 
        (|U\cap \Side{W}(u)| + |W\cap \Side{U}(u)| + |U\cap W|)
\intertext{Since $u\in \Side{W}(u)\cap \Side{U}(u)$, this is}
    &\leq 2\kappa + 2t - 1 - (\kappa+1)\\
    &\leq 2\ceil{\frac{n-\kappa}{2}} + \kappa - 2\\
    &< n,
\end{align*}
which contradicts the definition of $n=|V|$.  We conclude
that when $t\leq \ceil{(n-\kappa)/2}$, $\Small{t}(u)$ is unique
if it exists.
\end{proof}

\section{A Data Structure for $(\kappa+1)$-Connectivity Queries}\label{sect:data-structure}

In this section we design an efficient data structure that, given $u,v$,
answers $(\kappa+1)$-connectivity queries, i.e., reports that
$\kappa(u,v) = \kappa$
and produces a minimum $\kappa$-cut separating $u,v$,
or reports that $\kappa(u,v) \geq \kappa+1$.

We work with the mixed-cut definition of $\kappa(u,v)$ (see Remark~\ref{remark:kappa-def}), 
which is the minimum
size set of vertices and edges that need to be removed to disconnect $u$ and $v$,
or equivalently, the maximum size set of internally vertex-disjoint paths 
joining $u$ and $v$.\footnote{If $\{u,v\}\not\in E(G)$ and $\kappa(u,v)=\kappa$,
then there exists $U\subset V$, $|U|=\kappa$, 
such that removing $U$ disconnects 
$u,v$.  If $\{u,v\}\in E(G)$ then there exists $U\subset V$, $|U|=\kappa-1$,
such that removing $U$ and $\{u,v\}$ disconnects $u,v$.  
In this case the single-edge path $\{u,v\}$
would count for one of the $\kappa$ 
internally vertex disjoint paths, the other $\kappa-1$ 
passing through distinct vertices of $U$.}

\begin{theorem}\label{thm:data-structure}
Given a $\kappa$-connected graph $G$, we can construct in $\tilde{O}(m + \poly(\kappa)n)$ time
a data structure occupying $O(\kappa n)$ space that answers the following queries.
Given $u,v\in V(G)$, report whether $\kappa(u,v) = \kappa$ or $\geq \kappa+1$
in $O(1)$ time.  If $\kappa(u,v)=\kappa$, report a $\kappa$-cut separating $u,v$
in $O(\kappa)$ time.
\end{theorem}

\paragraph{Sparsification.}
In $O(m)$ time, the Nagamochi-Ibaraki~\cite{NagamochiI92} algorithm produces a subgraph $G'$ 
that has arboricity $\kappa+1$ and hence
at most $(\kappa+1)n$ edges, 
such that $\kappa_{G'}(u,v) = \kappa_{G}(u,v)$ whenever 
$\kappa_G(u,v)\leq \kappa+1$, and $\kappa_{G'}(u,v) \geq \kappa+1$ 
whenever $\kappa_G(u,v) \geq \kappa+1$.  
Without loss of generality we may assume
$G$ is the \emph{output} of the Nagamochi-Ibaraki 
algorithm.

\paragraph{The Data Structure.}
Throughout this section we fix the threshold $t=\ceil{\frac{n-\kappa}{2}}$.  
Define $\Small{}(u) = \Small{t}(u)$ to be the unique 
minimum $\kappa$-cut with $\Side{\Small{}(u)}(u)\leq t$, 
if any such cut exists, and $\Small{}(u)=\perp$ otherwise.
The data structure stores, for each $u\in V(G)$, 
$\Small{}(u), |\Side{\Small{}(u)}(u)|,$ 
a $O(\log n)$-bit identifier for $\Side{\Small{}(u)}(u)$,
and for each vertex $v\in N(u)\cap \Small{}(u)$,
a bit $b_{u,v}$ indicating 
whether $\{\{u,v\}\}\cup \Small{}(u)\backslash \{v\}$ is 
a mixed cut disconnecting $u$ and $v$.
Furthermore, when $|\Side{\Small{}(u)}(u)|\leq \kappa-1$,
we store $\Side{\Small{}(u)}(u)$ explicitly. 
When $\Small{}(u)=\perp$ we will say 
$\Side{\Small{}(u)}(u)=G$ and hence
$\abs{\Side{\Small{}(u)}}=n$.
The total space is $O(\kappa n)$.

\paragraph{Connectivity Queries.}
The query algorithm proceeds to the 
first applicable case. Note in the following, $\Small{}(u)$ may be $\perp$, and for all vertices $v$, we define $v\notin \perp$.
\begin{description}
\item[Case I: $\Small{}(u)=\Small{}(v)$ and $\Side{\Small{}(u)}(u)=\Side{\Small{}(v)}(v)$.]
Then $\kappa(u,v)\geq \kappa+1$. 
%
\item[Case II: $u\not\in \Small{}(v)$ and $v\not\in \Small{}(u)$.]
Then $\kappa(u,v)=\kappa$.
Without loss of generality suppose that $\abs{\Side{\Small{}(u)}(u)}\leq \abs{\Side{\Small{}(v)}(v)}$.  
Then $\Small{}(u)$ is a $\kappa$-cut separating $u$ and $v$.
\item[Case III: $v\in \Small{}(u)\cap N(u)$, or the reverse.] The bit $b_{u,v}$ indicates whether $\kappa(u,v)\geq \kappa+1$
or $\kappa(u,v)=\kappa$, in which case $\{\{u,v\}\}\cup \Small{}(u)\backslash\{v\}$ is the $\kappa$-cut. 
\item[Case IV: $v\in \Small{}(u), u\in \Small{}(v)$.] 
Then $\kappa(u,v)\geq \kappa+1$.
\item[Case V: $v \in \Small{}(u), u\not\in \Small{}(v)$, or the reverse.] 
If $\abs{\Side{\Small{}(v)}(v)} \leq \kappa - 1$, directly check whether $u\in \Side{\Small{}(v)}(v)$. If so 
then $\kappa(u, v)\geq \kappa + 1$; 
if not then $\Small{}(v)$ disconnects them.  Thus 
$\abs{\Side{\Small{}(v)}(v)} \ge \kappa$. If $\abs{\Side{\Small{}(v)}(v)} \leq \abs{\Side{\Small{}(u)}(u)}$
then $\Small{}(v)$ is 
a $\kappa$-cut separating $u$ and $v$, and otherwise 
$\kappa(u,v)\geq \kappa+1$. 
\end{description}

\medskip

Lemmas~\ref{lemma:minimal-cuts-laminar}, 
\ref{lemma:cut-by-minimal-cuts}, and Theorem~\ref{thm:query-correctness} 
establish the \emph{correctness}
of the query algorithm.  Its construction algorithm is described
and analyzed in Section~\ref{subsect:construction}.

\begin{lemma}\label{lemma:minimal-cuts-laminar}
If $v\in \Side{\Small{}(u)}(u)$,
then either $\Small{}(v)=\Small{}(u)$
or $\Small{}(v)$ is a laminar cut of $\Small{}(u)$ with
$\Side{\Small{}(v)}(v)\subset \Side{\Small{}(u)}(u)$.
\end{lemma}

\begin{proof}
$\Small{}(u)$ is (\Rmnum{3}, $t$)-small w.r.t.~$v$. By Theorem~\ref{thm:small-side-cut},
$\Small{}(v)$ exists and 
$\Side{\Small{}(v)}(v)\subseteq \Side{\Small{}(u)}(v)$. 
\end{proof}

\begin{lemma}\label{lemma:cut-by-minimal-cuts}
Suppose 
$u$ and $v$ are not $(\kappa+1)-$connected, i.e., 
$\kappa(u, v)=\kappa$.
If $\{u,v\}\not\in E(G)$,
then they are disconnected
by $\Small{}(u)$ or $\Small{}(v)$, 
and if $\{u,v\}\in E(G)$,
then they are disconnected by
$\{(u,v)\}\cup \Small{}(u)\backslash \{v\}$
or
$\{(u,v)\}\cup \Small{}(v)\backslash \{u\}$.
\end{lemma}

\begin{proof}
First suppose $\{u, v\}\notin E(G)$ and let $X$ be any cut separating
$u$ and $v$.  When $t=\ceil{\frac{n-\kappa}{2}}$ either 
$\abs{\Side{X}(u)} \leq t$ or $\abs{\Side{X}(v)}\leq t$.
W.l.o.g.~suppose it is the former, then $\Small{}(u)$ exists
and by Theorem~\ref{thm:small-side-cut}, $\Side{\Small{}(u)}(u)\subseteq \Side{X}(u)$,
so $\Small{}(u)$ also separates $u$ and $v$.

If $\{u, v\}\in E(G)$, suppose $(\kappa-1)$ vertices 
$W=\{w_1, w_2, \ldots, w_{\kappa-1}\}$
and $\{u, v\}$ disconnect $u$ and $v$. 
After removing $W$ from the graph, $G\backslash W$ is still connected. 
By deleting the edge $\{u, v\}$, the graph breaks into exactly two 
connected components, say $A$ and $B$ with 
$u\in A$ and $v\in B$. 
Then $W\cup \{u\}$ forms a $\kappa$-cut
with $\Side{W\cup\{u\}}(v)=B$, and $W\cup \{v\}$ also forms a $\kappa$-cut 
with $\Side{W\cup \{v\}}(u)=A$.
Clearly we have
\[
n = \abs{W} + \abs{A} + \abs{B} = \kappa-1 + \abs{A} + \abs{B}.\]
W.l.o.g.~suppose $\abs{A}\leq \abs{B}$, then 
\[
\abs{A} \leq \floor{\frac{n-\kappa+1}{2}} = \ceil{\frac{n-\kappa}{2}} = t.
\]
Thus $\Small{}(u)$ exists, 
$\Side{\Small{}(u)}(u)\subseteq \Side{W\cup \{v\}}(u)$, 
and $\Small{}(u)$ is either $W\cup\{v\}$ or 
a laminar cut of $W\cup\{v\}$ in side $A$. 
Since $\{u, v\}\in E(G)$, we have $v\in \Small{}(u)$. 
If we remove $\{u, v\}$ from $G$, then any path from $u$ to $v$ goes through 
a vertex in $W$, but any path from $u$ to a vertex in $W$ goes through 
a vertex in $\Small{}(u)\backslash \{v\}$. 
Therefore, $\{\{u, v\}\}\cup \Small{}(u)\backslash \{v\}$ is a mixed cut separating $u,v$
as it blocks all $u$-$v$ paths.
\end{proof}

\begin{theorem}\label{thm:query-correctness}
The query algorithm correctly answers
$(\kappa+1)$-connectivity queries.
\end{theorem}

\begin{proof}
Suppose the algorithm terminates in Case I.
It follows that $u\not\in\Small{}(v), v\not\in\Small{}(u)$, and neither
$\Small{}(u)$ nor $\Small{}(v)$ disconnect $u$ and $v$. 
Lemma~\ref{lemma:cut-by-minimal-cuts} implies that $\kappa(u,v)\geq \kappa+1$.

In Case II, if $\Small{}(u)\neq \perp$ 
but $\Small{}(v)=\perp$ then $\Small{}(u)$ is the cut 
separating $u,v$ and since 
$\abs{\Side{\Small{}(u)}(u)} < \abs{\Side{\Small{}(v)}(v)}=n$,
then the query is answered correctly.  
If both $\Small{}(u),\Small{}(v)\neq \perp$, 
then by Lemma~\ref{lemma:minimal-cuts-laminar}, 
$v\notin \Side{\Small{}(u)}(u)$ and once again the
query is answered correctly.

In Case III, by Lemma~\ref{lemma:cut-by-minimal-cuts}, if $u$ and $v$
are separated by a $\kappa$-cut, they are separated by
$\{\{u,v\}\}\cup \Small{}(u)\backslash\{v\}$
(if $\Small{}(u)\neq\perp$)
or 
$\{\{u,v\}\}\cup \Small{}(v)\backslash\{u\}$
(if $\Small{}(v)\neq \perp$), 
and this information is stored 
in the bit $b_{u,v},b_{v,u}$. 

If we get to Case IV then $\{u,v\}\not\in E(G)$ and neither
$\Small{}(u)$ nor $\Small{}(v)$ separate $u,v$, hence by 
Lemma~\ref{lemma:cut-by-minimal-cuts}, 
$\kappa(u,v)\geq \kappa+1$ and the query is answered correctly.

Case V is the most subtle.
Because $v\in \Small{}(u)$ and $\{u,v\}\not\in E(G)$, Lemma~\ref{lemma:cut-by-minimal-cuts} implies
that \underline{\emph{if}} $\kappa(u,v)=\kappa$, then 
$u,v$ must be separated by $\Small{}(v)$.
If $\Small{}(v)=\perp$ then $\kappa(u,v)\geq \kappa+1$ and 
the query is answered correctly.
If $|\Side{\Small{}(v)}(v)|\leq \kappa-1$ then the query
explicitly answers the query correctly by direct lookup.
Thus, we proceed under the assumption that
$\Small{}(v)\neq\perp$ exists and is not small.

If $u\in \Side{\Small{}(v)}(v)$ then 
$\Small{}(v)$ does not disconnect $u$ and $v$, 
and by Lemma~\ref{lemma:minimal-cuts-laminar}, 
$\abs{\Side{\Small{}(v)}(v)} > \abs{\Side{\Small{}(u)}(u)}$,
so the query is handled correctly in this case.

If $u\notin \Side{\Small{}(v)}(v)$ then $\Small{}(v)$
separates $u$ and $v$, so we must argue that
$\abs{\Side{\Small{}(v)}(v)}\leq \abs{\Side{\Small{}(u)}(u)}$ 
for the query algorithm 
to work correctly.
It cannot be that $\Small{}(v)$ and $\Small{}(u)$
have a laminar relation, so by Theorem~\ref{thm1}
they must have a crossing matching, wheel,
or small type relation.  If they have the small-type
relation then the small sides of $\Small{}(u)$ are
contained in $\Small{}(v)$ (contradicting $u\not\in \Small{}(v)$)
or the small sides of $\Small{}(v)$ are contained in $\Small{}(u)$,
but we have already rules out this case.
Thus, the remaining cases to consider are wheel 
and crossing matching type.

Suppose $\Small{}(u),\Small{}(v)$ form a 4-wheel
$(T; C_1,C_2,C_3,C_4)$.
Then $u\not\in \Small{}(v)$ 
appears in a sector of the wheel, say $S_1$.
Then $C(1,2)$ is a cut violating the minimality
of $\Small{}(u)=C(1,3)$.

Suppose $\Small{}(u),\Small{}(v)$ have a crossing
matching type relation. 
Let $A_1 = \Side{\Small{}(u)}(u)$ and $A_2$ be the other side of $\Small{}(u)$,
and $B_1 = \Side{\Small{}(v)}(v)$ and $B_2$ be the other side of $\Small{}(v)$.
Then $u\in A_1\cap B_2$, and it must be that the diagonal quadrant $A_2\cap B_1=\emptyset$.
Suppose otherwise, i.e., $A_2\cap B_1\neq \emptyset$, 
and let $X=(\Small{}(u)\cap B_2)\cup (\Small{}(v)\cap A_1)\cup (\Small{}(u)\cap \Small{}(v))$.
Then by Corollary~\ref{cor:CD}
$X$ is a $\kappa$-cut with $\Side{X}(u)=A_1\cap B_2$, contradicting the minimality of $\Small{}(u)$.
Thus, $\Small{}(v)$ is a crossing matching cut of $\Small{}(u)$ in side $A_2$ w.r.t.~
some $Q\subseteq \Small{}(u)\cap B_1$ with $v\in Q$.
By Theorem~\ref{thm1} 
and $u\in A_1\cap B_2\neq \emptyset$, we have
\[
|A_1\cap \Small{}(v)| = |B_2\cap \Small{}(u)| \geq |A_2\cap \Small{}(v)| = |B_1\cap \Small{}(u)|=|Q|.
\]
Thus,
\[
|\Side{\Small{}(u)}(u)|=|A_1| >|(A_1\cap B_1)\cup Q|=|\Side{\Small{}(v)}(v)|,
\]
establishing the correctness in the crossing matching case. (The strictness of the inequality is because $A_1\cap B_2\neq \emptyset$.)
\end{proof}

\subsection{Construction of the Data Structure}\label{subsect:construction}

We assume Nagamochi-Ibaraki sparsification~\cite{NagamochiI92} has already been applied,
so $G$ has arboricity $\kappa+1$ and $O(\kappa n)$ edges.
We use the recent Forster et al.~\cite{ForsterNYSY20} algorithm for computing 
the connectivity $\kappa=\kappa(G)$ in $\tilde{O}(\poly(\kappa) n)$ time and 
searching for $\kappa$-cuts.  

\begin{theorem}[Consequence of Forster, Nanongkai, Yang, Saranurak, and Yingchareonthawornchai~\cite{ForsterNYSY20}]\label{thm:Forstertal}
Given parameter $s\leq \ceil{\frac{n-\kappa}{2}}$ and a vertex $x$, 
there exists an algorithm that runs in time $O(s\kappa^3)$ with the following guarantee.
If the cut $\Small{s}(x)$ exists, it is reported with probability $\geq \frac{3}{4}$,
otherwise the algorithm returns $\perp$.
If the cut $\Small{s}(x)$ does not exist, the algorithm always returns $\perp$.
\end{theorem}

\begin{corollary}\label{cor:Forsteretal}
Given $x\in V(G)$ and an integer $s\leq \ceil{\frac{n-\kappa}{2}}$, we can, with high probability~$1-1/\poly(n)$, 
compute $\Small{s}(x)$ in $\tilde{O}(|\Side{\Small{s}(x)}(x)|\cdot \kappa^3)$ time, 
or determine that $\Small{s}(x)$ does not exist
in $\tilde{O}(s\kappa^3)$ time.
\end{corollary}

\begin{proof}
High probability bounds can be accomplished by 
repeating the algorithm of Theorem~\ref{thm:Forstertal}
$O(\log n)$ times.  Use a doubling search to find
$\Small{s_0}(x)$ for each $s_0 = 2^i$, 
$i\leq \ceil{\log s}$.  When $s_0 \geq |\Side{\Small{s}(x)}(x)|$ the procedure will succeed w.h.p.
\end{proof}

We call the procedure of Corollary~\ref{cor:Forsteretal} $\FindSmall(x, s)$.

\begin{definition}
  Let $U$ be a cut, $A$ a side of $U$.  We use the notation 
  $\overline{A}=G\backslash (U\cup A)$ to be the region of all other sides of $U$.
  Define $G(U, \overline{A})$ to be the graph induced by $U\cup \overline{A}$,
  supplemented with a $\kappa$-clique on $U$.
If $W$ is a cut in $G(U,\overline{A})$, define $\Side{W}^{G(U,\overline{A})}(x)$ to be $\Side{W}(x)$ in the graph $G(U,\overline{A})$.
\end{definition}

\begin{lemma}\label{lemma:laminar-in-region}
  Let $U$ be a $\kappa$-cut, $A$ be a side of $U$, and $W$ be a 
  set of $\kappa$ vertices in $G(U,\overline{A})$.
  Then $W$ is a $\kappa$-cut in $G(U, \overline{A})$
  if and only if $W$ is a laminar cut of $U$ in one of the sides of $\overline{A}$.
  Moreover, when $W$ is such a cut, for any vertex $u\in U\backslash W$,
  \[
  \Side{W}(u) = \Side{W}^{G(U, \overline{A})}(u) \cup A.
  \]
\end{lemma}

\begin{proof}
  A laminar cut $W$ of $U$ in one of the sides of $\overline{A}$
  does not disconnect vertices of $U$, so it is a cut in $G(U, \overline{A})$. When $W$ is a cut in $G(U, \overline{A})$, all vertices of $U\backslash W$ should be in one side $B$ of $W$. 
  Let the other sides of $W$ form a region $C$. 
  Then in $G$, by Lemma~\ref{lem1}, there exists paths from vertices in $A$ to $U\backslash W$ that are not blocked by $W$, so vertices in $A$ and $B$ together form a side $D$ of $W$. But any path from $C$ to $A$ should pass through some vertex in $U$, while any path from $C$ to $B$ (involving $U\backslash W$) should pass through $W$, 
  so $W$ is a cut in $G$. This proves the first statement.
  
  Note $B$ here is exactly $\Side{W}^{G(U, \overline{A})(u)}$, and $D$ is exactly $\Side{W}(u)$, and we have that $D=B\cup A$. This proves the second statement.
\end{proof}

Lemma~\ref{lem:expand} shows how, beginning with a cut $X$ 
where $\Side{X}(u)$ is small, can find another cut $Y$ 
(if one exists) where $\Side{Y}(u)$ is about $M$, 
in $\tilde{O}(M\kappa^4)$ time.  The difficulty is that there could 
be an unbounded number of cuts ``between'' $X$ and $Y$ that would prevent
the algorithm of Theorem~\ref{thm:Forstertal} from finding $Y$ directly.

\begin{lemma}\label{lem:expand}
  Fix any integer $M\leq t/2$, vertex $u$, and
  cut $X$ with $A = \Side{X}(u), |A| \leq 2M$.
  There exists an algorithm $\Expand(u,A,M)$ that runs in time 
  $\tilde{O}(M \kappa^4)$ and, w.h.p., 
  returns a cut $Y$ satisfying the following properties.
\begin{itemize}
    \item $\Side{X}(u)\subseteq \Side{Y}(u)$.
    \item $\abs{\Side{Y}(u)}\leq 2M$.
    \item If there exists a cut $Z$ that is (\Rmnum{3}, $M$)-small 
    w.r.t.~$u$, then $\abs{\Side{Y}(u)}\geq \abs{\Side{Z}(u)}$.
\end{itemize}
\end{lemma}

\begin{proof}
  This algorithm uses Corollary~\ref{cor:Forsteretal}
  and Lemma~\ref{lemma:laminar-in-region}.
  \begin{enumerate}
      \item Initially $Y\gets X$. While $\abs{\Side{Y}(u)} < M$,
      \begin{enumerate}
        \item For each vertex $v\in Y$, in parallel,
              \begin{enumerate}
                  \item In the graph 
                  $G(Y, \overline{\Side{Y}(u)})$, 
                  run $\FindSmall(v, M)$.
              \end{enumerate}
          \item The moment any call to $\FindSmall$ halts in step 
          (i) with a cut $W$, stop all such calls and set $Y\gets W$.  If all $|Y|$ calls to $\FindSmall$ run to completion without finding a cut, halt and return $Y$.
      \end{enumerate}
  \end{enumerate}

  Throughout the algorithm, $Y$ is always a cut such that $\Side{X}(u)\subseteq \Side{Y}(u)$ and $\abs{\Side{Y}(u)}\leq 2M$. Furthermore, $\abs{\Side{Y}(u)}$ is strictly increasing, 
  so the algorithm terminates. 
  By Corollary~\ref{cor:Forsteretal} the running time 
  of (a) is $\tilde{O}(\Delta\kappa^4)$ where 
  $\Delta = |\Side{W}(u)| - |\Side{Y}(u)|$ if a cut $W$
  is found, and $\Delta=M$ otherwise.  (The extra factor $\kappa$
  is due to the parallel search in step (a).)
  The sum of the $\Delta$s telescopes to $O(M)$, so 
  the overall running time is $\tilde{O}(M\kappa^4)$.

  Suppose the cut $Z$ exists.  If it were not the case that $\abs{\Side{Y}(u)} \geq \abs{\Side{Z}(u)}$, then in the last iteration,
  $\abs{\Side{Y}(u)} < \abs{\Side{Z}(u)}\leq M$.
  We argue below that there is another cut to find, and therefore 
  that the probability the algorithm terminates prematurely 
  is $1/\poly(n)$, by Corollary~\ref{cor:Forsteretal}.

  Note that $\abs{\Side{Y}(u)} + \abs{\Side{Z}(u)} \leq 2M \leq t$, so $C\bydef G\backslash (\Side{Y}(u)\cup \Side{Z}(u)\cup Y\cup Z)\neq \emptyset$. Apply Corollary~\ref{cor:CD} to sets $C$ and $D=\{u\}$, 
  we may set $U=\MinCut{C; D}$. Because $\Region{U}(C)\subseteq \Region{Y}(C)\cap \Region{Z}(C)=C$, $C$ is actually a region of $U$, so $U$ is a laminar cut of $Y$, $\Side{U}(u)\subseteq \Side{Y}(u)\cup \Side{Z}(u)$. 
  Since $\abs{\Side{Y}(u)}<\abs{\Side{Z}(u)}$, $U\neq Y$. 
  Then for any vertex $v\in Y\backslash U$, 
  \[
  0<\abs{\Side{U}^{G(Y, \overline{\Side{Y}(u)})}(v)}\leq \abs{\Side{U}(u)}-\abs{\Side{Y}(u)}\leq \abs{\Side{Z}(u)}\leq M.
  \]
  Therefore, $U$ or some other cut should have been found in step (a)
  in the last iteration.
\end{proof}

\begin{theorem}[Consequence of Picard and Queyrenne~\cite{PicardQ80}]\label{thm:Picard-Queyrenne}
Let $H=(V,E,\operatorname{cap})$ be a capacitated 
$s$-$t$ flow network and $f$ be a maximum flow.
In $O(|E|)$ time we can compute a directed acyclic 
graph $H'=(V',E')$ with $|E'|\leq |E|$, 
and an embedding $\phi : V\to V'\cup\{\perp\}$, 
such that the downward closed sets of $H'$ 
are in 1-1 correspondence with the minimum $s$-$t$ cuts of $H$.
I.e., if $V''\subseteq V'$ is a vertex set with no arc
of $E'$ leaving $V''$, then 
$(\phi^{-1}(V''), \overline{\phi^{-1}(V'')})$ 
is a min $s$-$t$ cut in $H$, and all min $s$-$t$ cuts in $H$ 
can be expressed in this way.  Here $\phi(v)=\perp$
if $v$ appears on the ``$t$'' side of every min $s$-$t$ cut.
\end{theorem}

We use 
Corollary~\ref{cor:result-of-picard-queyrenne} to find $\Small{t}(u)$ for potentially many vertices 
$u$ in bulk.

\begin{corollary}\label{cor:result-of-picard-queyrenne}
Fix two disjoint, non-empty vertex sets $C$ and $D$.
In $O(\kappa^2(n-|C|-|D|))$ time, we can output a cut $S(v)$ for every $v\in V\backslash(C\cup D)$,
such that if $\Small{}(v)$ exists and 
$C\subseteq \Side{\Small{}(v)}(v)$, then $S(v)=\Small{}(v)$.
\end{corollary}

\begin{proof}
Form a flow network $\vec{G}$ as in the proof of Corollary~\ref{cor:CD},
shrinking $C$ and $D$ to vertices $s$ and $t$, splitting each $v$ into $v_{\operatorname{in}},v_{\operatorname{out}}$, etc.
The minimum $s_{\operatorname{out}}$-$t_{\operatorname{in}}$ 
cuts in $\vec{G}$ are in 1-1 correspondence with the 
minimum vertex cuts separating $C,D$ in $G$.  We are only interested in these
cuts if they have size $\kappa$, so in this case a maximum flow $f$ can be computed in
$O(\kappa^2(n-|C|-|D|))$ time.  (Recall that the graph is assumed to have
arboricity $\kappa+1$, so every induced subgraph of $\vec{G}$ 
has density $O(\kappa)$.)

Given $\vec{G},f$, the Picard-Queyrenne representation $\vec{G}',\phi$ can be computed
in linear time.  For each $v\in G$, let $S(v)$ be the cut corresponding 
to the smallest downward-closed set containing 
$\phi(v)$ in $\vec{G}'$.
Thus, the set $\{S(v)\}_{v\in V\backslash(C\cup D)}$ can be enumerated 
by a bottom-up traversal of $\vec{G}'$ in 
$O(|E(\vec{G}')|)=O(\kappa(n-|C|-|D|))$ time.
If $C\subseteq \Side{\Small{}(v)}(v)$ then clearly $S(v)=\Small{}(v)$.
\end{proof}

\medskip

We are now ready to present the entire construction algorithm.

\paragraph{Preamble.} The algorithm maintains some $\kappa$-cut $T(u)$ for each $u$, which is initially $\perp$, and stores
$|\Side{T(u)}(u)|$.  If the algorithm makes no errors,
$T(u)=\Small{t}(u)=\Small{}(u)$ at the end of the computation.
The procedure $\Update(u,U)$ updates $T(u)\gets U$ if $U$ is a better cut, i.e., $|\Side{U}(u)| \leq \min\{|\Side{T(u)}(u)|-1,t\}$,
and does nothing otherwise.  $\Update(A,U)$ is short for $\Update(u,U)$ for all $u\in A$.

\paragraph{Step 1: very small cuts.} For each $u\in V$, 
let $U \gets \FindSmall (u, 100\kappa)$ and then $\Update (u,U)$.
This takes $\tilde{O}(n\kappa^4)$ time.

\paragraph{Step 2: unbalanced cuts.} For each index $i$
such that $100\kappa < 2^i \leq t$, let $\alpha=2^i$,
and pick a uniform sample 
$V_i\subset V$ of size $(n\log n)/\alpha$.\footnote{The Forster et al.~\cite{ForsterNYSY20} algorithm samples vertices 
proportional to their degree.  Note that after the Nagamochi-Ibaraki~\cite{NagamochiI92} sparsification,
the minimum degree is at least $\kappa$ and the density 
of every induced subgraph is at most $\kappa+1$, 
so it is equally effective to do vertex sampling.}
For each $u\in V_i$, compute $\Small{\alpha}(u) \gets \FindSmall(u,\alpha)$. 
If $\Small{\alpha}(u)=\Small{}(u)\neq \perp$, 
we first do an $\Update(\Side{\Small{}(u)}(u),\Small{}(u))$,
then compute $Y \gets \Expand(u,\Small{}(u),\alpha)$.
For each $v\in Y$ we compute $W_v \gets \FindSmall(w,\alpha)$
and then $\Update(v,W_v)$.  
We then run the algorithm of Corollary~\ref{cor:result-of-picard-queyrenne}
with $C=\Side{\Small{}(u)}(u)$ and $D=V\backslash (Y\cup \Side{Y}(u))$, which returns a set of cuts $\{S(v)\}_{v\in V\backslash(C\cup D)}$.  For each such $v\in \Side{Y}(u)\backslash \Side{\Small{}(u)}(u)$, do an $\Update(v,S(v))$.
For each index $i$, the running time is 
$\tilde{O}(|V_i|\cdot \alpha\kappa^4) = \tilde{O}(n\kappa^4)$,
which is $\tilde{O}(n\kappa^4)$ overall.

\paragraph{Step 3: balanced cuts.}
Sample $O(\log n)$ pairs $(x,y)\in V^2$.
For each such pair, compute $U\gets \FindSmall(x,t)$.
If $U\neq \perp$, apply the algorithm of Corollary~\ref{cor:result-of-picard-queyrenne}
to $C=\Side{\Small{}(x)}(x)$ and $D=\{y\}$,
which returns a set $\{S(v)\}$.  Then do
an $\Update(v,S(v))$ for every $v\in V\backslash (C\cup D)$.
By Corollary~\ref{cor:Forsteretal} 
this takes $\tilde{O}(\kappa^3 n)$ time.

\paragraph{Step 4: adjacent vertices.}
At this point it should be the case that 
$T(u) = \Small{}(u)$ for all $u$.  
For each $v\in T(u)\cap N(u)$ compute and 
set the bit $b_{u,v}$.
(This information can be extracted
from the calls to $\FindSmall$ and the 
algorithm of Corollary~\ref{cor:result-of-picard-queyrenne}
in the same time bounds.)

Lemma~\ref{lem:cut-capture} 
is critical to proving the correctness
of the algorithm's search strategy.

\begin{lemma}\label{lem:cut-capture}
Suppose $\Small{}(u)$ and $\Small{}(v)$ exists, $u\in \Side{\Small{}(v)}(v)$, and 
suppose there is a cut $W$ such that 
\[
\kappa \leq \abs{\Side{\Small{}(v)}(v)} < \abs{\Side{W}(u)} \leq t.
\]
Then $v\in W\cup \Side{W}(u)$.
\end{lemma}

\begin{proof}
  Because $u\in \Side{\Small{}(v)}(v)$, by Lemma~\ref{lemma:minimal-cuts-laminar}, $\Small{}(u)\neq \perp$ is a laminar cut of $\Small{}(v)$ and 
  $\Side{\Small{}(u)}(u)\subseteq \Side{\Small{}(v)}(u)$.
  
  Consider the relationship between $W$ and $U=\Small{}(v)$. 
  We use the same notation $U_*$, $W_*$, $A_*$, $B_*$ and $T$ 
  from Theorem~\ref{thm1}. 
  Neither $U$ nor $W$ is (\Rmnum{1}, $\kappa-1$)-small, 
  so they may have laminar, wheel, or crossing matching type relation. 
  
  If they have laminar type relation, then 
  $u\in \Side{W}(u)\cap \Side{U}(u)$, so either $\Side{W}(u)$ is $A_{i^*}$ or $\Side{U}(u)$ is $B_{j^*}$, 
  but they cannot be both true as otherwise
  \[\abs{V}=n>2\ceil{\frac{n-\kappa}{2}}+\kappa - 2\geq \abs{A_{i^*}}+\abs{B_{j^*}}-\abs{A_{i^*}\cap B_{j^*}} + \abs{T}=\abs{V}.\]
  If $\Side{W}(u)$ is $A_{i^*}$ then $v\in \Side{W}(u)$. If $\Side{U}(u)$ is $B_{j^*}$ then $\abs{\Side{W}(u)}<\abs{\Side{\Small{}(v)}(u)}$, 
  which contradicts the definition of $W$.
  
  If they have wheel type relation, then $v$ cannot lie in a sector of the $4$-wheel formed by $W$ and $U$, 
  as this would violate the minimality of $\Small{}(v)$.
  Therefore $v \in W$.
  
  If they have crossing matching type relation, suppose $A_1\cap B_1\neq \emptyset$, $A_2\cap B_2\neq \emptyset$ and $A_1\cap B_2=\emptyset$. Then $v\in W$ (and the proof is done) 
  or $v\in A_2\cap B_1$, because for $v\in A_1\cap B_1$, $W_1\cup T\cup U_1$ is a cut with side smaller than $\Small{}(v)$, and for $v\in A_2\cap B_2$, $U_2\cup T\cup W_2$ is a cut with a side smaller than $\Small{}(v)$.
  
  If $v\in A_2\cap B_1$, then by Theorem~\ref{thm1}, $\abs{W_2}\geq \abs{U_2}$.
  Also, $A_2$ is $\Side{\Small{}(v)}(v)$ so $u\in A_2$. Thus
  \[\abs{\Side{W}(u)}-\abs{\Side{U}(v)}=\abs{B_2}-\abs{A_2} =\abs{U_2}-\abs{W_2}-\abs{A_2\cap B_1}<0,\]
  which contradicts the condition that $\abs{\Side{W}(u)}>\abs{\Side{\Small{}(v)}(v)}$.
\end{proof}

\begin{theorem}
The construction algorithm correctly computes $\{\Small{}(v)\}_{v\in V}$ and runs in time $\tilde{O}(n\kappa^4)$.
\end{theorem}

\begin{proof}
  If $|\Side{\Small{}(v)}(v)|\leq 100\kappa$, 
  then $T(v) = \Small{}(v)$
  after Step 1, with high probability.
  
  Suppose that $|\Side{\Small{}(v)}(v)| \in [2^j,2^{j+1}]$ and $2^{j+1}\leq t$.
    Then with high probability, at least one vertex $x\in V_j$
    is sampled in Step 2 such that $x\in \Side{\Small{}(v)}(v)$.
    Step 2 ($\Expand$) computes a cut $Y$ such that
    $|\Side{Y}(x)|\geq |\Side{\Small{}(v)}(v)|$, 
    so by Lemma~\ref{lem:cut-capture}, either 
    $v\in \Side{Y}(x)$ or $v\in Y$.  In the former case
    $\Small{}(v)$ is computed using the Corollary~\ref{cor:result-of-picard-queyrenne} algorithm.
    In the latter case $\Small{}(v)$ is computed directly
    using $\FindSmall$.
  
  Finally, if $\Small{}(v)$ is balanced, say $|\Side{\Small{}(v)}(v)|\geq t/4$, then w.h.p.~we would
  pick a pair $(x,y)$ in Step 3 such that $x\in \Side{\Small{}(v)}(v)$
  and $y\in V\backslash (\Small{}(v)\cup \Side{\Small{}(v)}(v))$.
  If this holds the algorithm of Corollary~\ref{cor:result-of-picard-queyrenne} 
  correctly computes $\Small{}(v)$.
\end{proof}

\section{Conclusion}\label{sect:conclusion}

This paper was directly inspired by the extended abstract of Cohen, Di Battista, Kanevsky, and Tamassia~\cite{cohen1993reinventing}.
Our goal was to substantiate the main claims of this paper,
and to simplify and improve 
the data structure that answers 
$(\kappa+1)$-connectivity queries.

We believe that our structural theorems can, ultimately,
be used to develop even 
more versatile vertex-cut data structures.  For example, is it possible to answer the following 
more general queries in $O(\kappa)$ time using $O(\kappa n)$ space?
\begin{description}
\item[\textsf{Is-it-a-cut?}$(u_1,\ldots,u_\kappa)$:] Return \emph{true} iff $\{u_1,\ldots,u_\kappa\}$ forms a $\kappa$-cut.
\item[\textsf{Part-of-a-cut?}$(u_1,\ldots,u_{\kappa-g})$:]
Return \emph{true} iff the input can be extended to 
a $\kappa$-cut 
$\{u_1,\ldots,u_{\kappa-g}\}\cup \{u_{\kappa-g+1},\ldots,u_\kappa\}$.
\end{description}

We assumed throughout the paper that $\kappa$ was not too large, specifically $\kappa<n/4$.
When $n<2\kappa$, \emph{all} cuts are (I,$\kappa$)-small by our classification, and the classification theorem (Theorem~\ref{thm1}) says very little about the structure of such cuts.
Understanding the structure of minimum vertex cuts when $\kappa$ is large, 
relative to $n$, 
is an interesting open problem.


\end{document}